\documentclass[a4paper,UKenglish,cleveref, thm-restate]{lipics-v2021}
\pdfoutput=1 %
\newcommand{\fscdarxiv}[2]{#2} %

\fscdarxiv{}{\hideLIPIcs}

\title{Substitution for Non-Wellfounded Syntax with Binders through Monoidal Categories} %

\author{Ralph Matthes}{IRIT, Université de Toulouse, CNRS, Toulouse INP, UT3, Toulouse, France \and\url{https://www.irit.fr/~Ralph.Matthes/}}{Ralph.Matthes@irit.fr}{https://orcid.org/0000-0002-7299-2411}{}

\author{Kobe Wullaert}{Delft University of Technology, Netherlands \and\url{https://kfwullaert.github.io/}}{K.F.Wullaert@tudelft.nl}{https://orcid.org/0000-0003-4281-2739}{}

\author{Benedikt Ahrens}{Delft University of Technology, Netherlands \and University of Birmingham, United Kingdom \and \url{https://benediktahrens.gitlab.io} }{B.P.Ahrens@tudelft.nl}{https://orcid.org/0000-0002-6786-4538}{}

\authorrunning{R. Matthes, K. Wullaert, and B. Ahrens} %

\Copyright{Ralph Matthes, Kobe Wullaert, and Benedikt Ahrens} %

\ccsdesc[500]{Theory of computation~Type theory}
\ccsdesc[500]{Theory of computation~Logic and verification}

\keywords{Non-wellfounded syntax, Substitution, Monoidal categories, Actegories, Tensorial strength,  Proof assistant Coq, UniMath library} %

\fscdarxiv{\relatedversion{A full version that includes Appendix B and Appendix D.}}{} %
\fscdarxiv{\relatedversiondetails[cite=DBLP:journals/corr/abs-2308-05485]{Full Version}{https://arxiv.org/abs/2308.05485}}{} %

\acknowledgements{
  We thank Henning Basold for pointing us to the work on completely iterative algebras, leading to a simpler proof of \cref{thm:mhssfromfinal}. We also thank Thomas Lamiaux for valuable comments on a draft of this paper.
  We gratefully acknowledge the work by the Coq development team in providing the Coq proof assistant and surrounding infrastructure, as well as their support in keeping UniMath compatible with Coq. Not least, we thank the anonymous FSCD reviewers for their thoughtful feedback on our submission.
}%

\nolinenumbers %

\usepackage{proof}

\usepackage{mathtools}

\usepackage{xifthen}

\newcommand{\longhash}{7432feea2113a460eb5a69fbbba5fda02e2bf234}
\newcommand{\shorthash}{7432fee}

\newcommand{\coqdocbasebaseurl}{https://benediktahrens.gitlab.io/unimathdoc3/\shorthash}

\newcommand{\coqdocbaseurl}{\coqdocbasebaseurl /UniMath.}
\newcommand{\urlhash}{\#}
\newcommand{\coqdocurl}[2]{\coqdocbaseurl #1.html\urlhash #2}
\newcommand{\nolinkcoqident}[1]{\nolinkurl{#1}} %
\makeatletter
\newcommand{\coqident}{\begingroup\@makeother\#\@coqident}
\newcommand{\@coqident}[3][]{%
  \ifthenelse{\isempty{#2}}%
  {\nolinkcoqident{#3}}%
  {\ifthenelse{\isempty{#1}}%
  {\href{\coqdocurl{#2}{#3}}{\nolinkcoqident{#3}}}%
  {\href{\coqdocurl{#2}{#3}}{\nolinkcoqident{#1}}}}%
\endgroup}
\newcommand{\coqfile}[2]{%
  \ifthenelse{\isempty{#1}}%
  {\href{\coqdocbaseurl #2.html}{\nolinkcoqident{#2.v}}}%
  {\href{\coqdocbaseurl #1.#2.html}{\nolinkcoqident{#2.v}}}}
\makeatother

\usepackage[draft]{optional}
\newcommand{\plan}[1]{}
\newcommand{\BA}[1]{}
\newcommand{\RM}[1]{}
\newcommand{\KW}[1]{}
\opt{draft}{
\renewcommand{\plan}[1]{\textcolor{blue}{#1}\PackageWarning{TODO}{TODO: #1}}
\renewcommand{\BA}[1]{\textcolor{orange}{BA: #1}\PackageWarning{TODO}{TODO: #1}}
\renewcommand{\RM}[1]{\textcolor{purple}{RM: #1}\PackageWarning{TODO}{TODO: #1}}
\renewcommand{\KW}[1]{\textcolor{magenta}{KW: #1}\PackageWarning{TODO}{TODO: #1}}
}

\usepackage{tikz-cd}
\usepackage{xspace}

\newcommand{\cfont}[1]{\ensuremath{\mathsf{#1}}}

\newcommand{\catfont}[1]{\ensuremath{\mathcal{#1}}}

\newcommand{\CC}{\catfont{C}}
\newcommand{\DD}{\catfont{D}}
\newcommand{\EE}{\catfont{E}}

\newcommand{\tensor}{\otimes}

\newcommand{\actor}{\mathsf{act}}

\newcommand{\catb}[1]{\mathbf{#1}}
\newcommand{\SET}{\catb{Set}}

\newcommand{\cf}{cf.\xspace}
\newcommand{\eg}{e.\,g.\xspace}

\newcommand{\ie}{i.\,e.\xspace}

\newtheorem{construction}[theorem]{Construction}

\DeclareFontFamily{U}{min}{}
\DeclareFontShape{U}{min}{m}{n}{<-> udmj30}{}

\usepackage{textgreek}
\usepackage{stmaryrd}
\usepackage{pmboxdraw}
\usepackage[pdf,all,2cell]{xy}\SelectTips{cm}{10}\SilentMatrices\UseAllTwocells
\usetikzlibrary{arrows,matrix,decorations.pathmorphing,
  decorations.markings, calc, backgrounds}
\usepackage{listings}
\usepackage{xifthen} %
\usepackage{microtype}

\newcommand{\teletype}[1]{\ensuremath{\mathtt{#1}}}
\newcommand{\systemname}[1]{\teletype{\color{darkgray}#1}\xspace}

\newcommand{\UniMath}{\systemname{UniMath}}

\newcommand{\Coq}{\systemname{Coq}}

\newcommand{\VV}{\mathcal{V}}
\newcommand{\WW}{\mathcal{W}}

\newcommand{\Id}[1]{\cfont{Id}}

\newcommand{\idwt}[1]{\cfont{id}}
\newcommand{\hcomp}{\mathbin{\cdot}} %

\newcommand{\lam}{\cfont{lam}}
\newcommand{\app}{\cfont{app}}

\newcommand{\set}{\ensuremath{\mathsf{Set}}\xspace}

\newcommand{\arity}{\cfont{ar}}

\newcommand{\sort}{\cfont{S}}
\newcommand{\ccsort}{\CC^{\sort}}
\newcommand{\bool}{\cfont{bool}}

\newcommand{\tinl}{\cfont{inl}}
\newcommand{\tinr}{\cfont{inr}}

\newcommand{\proj}[1]{\cfont{pr}_{#1}}

\newcommand{\option}{\cfont{option}}

\newcommand{\convert}{\ensuremath{\equiv}}
\newcommand{\eqdef}{\ensuremath{:\convert}}

\makeatletter
\def\prd#1{\@ifnextchar\bgroup{\prd@parens{#1}}{%
    \@ifnextchar\sm{\prd@parens{#1}\@eatsm}{%
    \@ifnextchar\prd{\prd@parens{#1}\@eatprd}{%
    \@ifnextchar\;{\prd@parens{#1}\@eatsemicolonspace}{%
    \@ifnextchar\\{\prd@parens{#1}\@eatlinebreak}{%
    \@ifnextchar\narrowbreak{\prd@parens{#1}\@eatnarrowbreak}{%
      \prd@noparens{#1}}}}}}}}
\def\prd@parens#1{\@ifnextchar\bgroup%
  {\mathchoice{\@dprd{#1}}{\@tprd{#1}}{\@tprd{#1}}{\@tprd{#1}}\prd@parens}%
  {\@ifnextchar\sm%
    {\mathchoice{\@dprd{#1}}{\@tprd{#1}}{\@tprd{#1}}{\@tprd{#1}}\@eatsm}%
    {\mathchoice{\@dprd{#1}}{\@tprd{#1}}{\@tprd{#1}}{\@tprd{#1}}}}}
\def\@eatsm\sm{\sm@parens}
\def\prd@noparens#1{\mathchoice{\@dprd@noparens{#1}}{\@tprd{#1}}{\@tprd{#1}}{\@tprd{#1}}}
\def\lprd#1{\@ifnextchar\bgroup{\@lprd{#1}\lprd}{\@@lprd{#1}}}
\def\@lprd#1{\mathchoice{{\textstyle\prod}}{\prod}{\prod}{\prod}({\textstyle #1})\;}
\def\@@lprd#1{\mathchoice{{\textstyle\prod}}{\prod}{\prod}{\prod}({\textstyle #1}),\ }
\def\tprd#1{\@tprd{#1}\@ifnextchar\bgroup{\tprd}{}}
\def\@tprd#1{\mathchoice{{\textstyle\prod_{(#1)}}}{\prod_{(#1)}}{\prod_{(#1)}}{\prod_{(#1)}}}
\def\dprd#1{\@dprd{#1}\@ifnextchar\bgroup{\dprd}{}}
\def\@dprd#1{\prod_{(#1)}\,}
\def\@dprd@noparens#1{\prod_{#1}\,}

\def\@eatnarrowbreak\narrowbreak{%
  \@ifnextchar\prd{\narrowbreak\@eatprd}{%
    \@ifnextchar\sm{\narrowbreak\@eatsm}{%
      \narrowbreak}}}
\def\@eatlinebreak\\{%
  \@ifnextchar\prd{\\\@eatprd}{%
    \@ifnextchar\sm{\\\@eatsm}{%
      \\}}}
\def\@eatsemicolonspace\;{%
  \@ifnextchar\prd{\;\@eatprd}{%
    \@ifnextchar\sm{\;\@eatsm}{%
      \;}}}

\newcommand{\compose}[2]{\ensuremath{{#1}\cdot{#2}}}
\newcommand{\defeq}{\coloneqq}

\newcommand{\pointforv}{\textit{pv}}
\newcommand{\pointforw}{\textit{pw}}
\newcommand{\ptd}{\mathit{ptd}}

\newcommand{\gst}[2]{\llparenthesis {#2} \rrparenthesis_{#1}}

\newcommand{\tout}{\cfont{out}}
\newcommand{\teqm}{\cfont{eqm}}

\newcommand{\cosign}{{\textit{co}}}
\newcommand{\ns}[2]{#1+\cdots+#2} %
\newcommand{\tuple}[1]{\langle #1 \rangle}

\newcommand{\seqt}[3]{#1\vdash #2:#3}
\newcommand{\Threet}{\mathsf{THREE}}

\newcommand{\impl}{\to}
\newcommand{\sortv}{\cfont v}
\newcommand{\sortt}{\cfont t}
\newcommand{\sorte}{\cfont e}
\newcommand{\nat}{\mathbb{N}}
\newcommand{\atom}{\cfont{atom}}
\newcommand{\atotype}{\cfont{atotype}}

\EventEditors{Jakob Rehof}
\EventNoEds{1}
\EventLongTitle{9th International Conference on Formal Structures for Computation and Deduction (FSCD 2024)}
\EventShortTitle{FSCD 2024}
\EventAcronym{FSCD}
\EventYear{2024}
\EventDate{July 10-13, 2024}
\EventLocation{Tallinn, Estonia}
\EventLogo{}
\SeriesVolume{299}
\ArticleNo{22}

\begin{document}

\maketitle

\begin{abstract}
  We describe a generic construction of non-wellfounded syntax involving variable binding and its monadic substitution operation.

  Our  construction of the syntax and its substitution takes place in category theory, notably by using monoidal categories and strong functors between them.
  A language is specified by a multi-sorted binding signature, say $\Sigma$.
  First, we provide sufficient criteria for $\Sigma$ to generate a language of possibly infinite terms, through $\omega$-continuity.
  Second, we construct a monadic substitution operation for the language generated by $\Sigma$.
  A cornerstone in this construction is a mild generalization of the notion of heterogeneous substitution systems developed by Matthes and Uustalu;
  such a system encapsulates the necessary corecursion scheme for implementing substitution.
  
  The results are formalized in the Coq proof assistant, through the
  UniMath library of univalent mathematics.
\end{abstract}

\section{Introduction}\label{sec:intro}

\subsection{General Motivation for Non-Wellfounded Syntax With Binders}
Non-wellfounded syntax with binders appears in its purest form in the coinductive reading of untyped $\lambda$-calculus. \emph{Potentially} non-wellfounded $\lambda$-terms still consist of variables, $\lambda$-abstractions and applications only, but the construction process with these constructors can go on forever.
Such construction processes can be described through functional programming, and the host programming language then serves as a meta-language for the description of those infinitary $\lambda$-terms. %
Instead of taking a programming perspective, one can also ask if a possibly circular definition of such a non-wellfounded term is well-formed, in the sense that it uniquely determines such a structure.
Naturally, uniqueness is understood up to bisimilarity, \ie, two such non-wellfounded $\lambda$-terms are considered equal if their infinite unfoldings have the same labels (indicating the applied constructor) in the same order on each level, starting at the root.
The presence of variable binding presents the extra challenge of having to consider this bisimilarity modulo renaming of bound variables, \ie, $\alpha$-equivalence -- a challenge that is amplified by the possibility of having an infinite number of bindings in a non-wellfounded $\lambda$-term.
In this paper, we either work on an abstract level that does not reveal this challenge, or we resort to a representation using nested datatypes that is a form of de Bruijn representation with well-scopedness guaranteed by the typing system (see, \eg, \cite{DBLP:journals/jfp/AllaisACMM21}), and therefore $\alpha$-equivalence is just not needed.

There are many uses of coinductive untyped $\lambda$-terms (such as Böhm trees), and coinductive readings of term structures with binding (also with simple types, \eg, as in an automata-theoretic analysis by Melliès \cite{DBLP:conf/lics/Mellies17}) have a counterpart in infinitary rewriting.

\subsection{A Motivational Application Scenario}\label{sec:appscenario}
We have an application scenario in mind for which the ``static'' part, \ie, the well-typed syntax itself, is important, even without the aforementioned ``dynamics'' of infinitary rewriting. It is more complicated than just $\lambda$-calculus, notably by the presence of embedded inductive types.
This in particular motivates our search for \emph{datatype-generic} constructions for a wide range of non-wellfounded simply-typed syntax.

The application scenario is as follows:
We want to represent the entire search space for inhabitants in simply-typed $\lambda$-calculus (STLC) by a potentially non-wellfounded term of a suitable calculus.
The inhabitation problem itself is the following question: ``Given a context $\Gamma$ and a type $A$ of STLC, is there a term $t$ of STLC such that $\Gamma\vdash t:A$?''. Taking into account the entire search space means including infinite runs that arise in a (naive) search loop. And the term of the ``suitable calculus'' should again have type $A$ in context $\Gamma$ but represent the search space and not only be a single inhabitant.
This calculus is informally given by the following grammar:%
\[
\begin{array}{lcrcl}
\textrm{(terms)} &  & N & ::=_\cosign & \lambda x^A.N\,|\, \ns{E_1}{E_n}\\
\textrm{(elimination alternatives)} &  & E & ::=_\cosign & x \tuple{N_1,\ldots,N_k}%
\end{array}
\]
with one constructor for each $n,k\geq0$, hence we have sums with any finite number of summands and tuples with any finite number of arguments.
We write $x$ in place of $x \tuple{}$---this captures $k=0$.
The elimination alternatives resemble the neutral terms of $\lambda$-calculus of the form $xN_1\ldots N_k$ -- we are only searching for inhabitants in normal form. They have this name because they correspond to repeated implication elimination (as expressed by STLC typing) and they are summands in $\ns{E_1}{E_n}$ that indicate a finite choice between those ``alternative'' $n$ summands.
Search for normal forms in STLC only has finitely many options at each choice point, even though, \eg, there are infinitely many inhabitants of the type of Church numerals.

The elements of the syntactic category of terms are also called ``forests''. The index $\cosign$ means that the grammar is read coinductively. There are two clauses that embed (finite) lists into the codata type. It therefore presents at least the challenges of non-wellfounded ``rose trees'', \ie, finitely-branching unlabeled trees without a bound on the branching width.
The scenario comes from \cite[Section 3.2]{EspiritoSantoMatthesPintoCoinductiveApproachAPAL}%
, and we plan to study it with our formalization.
The typing rules for these expressions are given in \cref{fig:typingforests}, with $\Gamma$ ranging over finite(!) typing contexts.
They are the usual implication introduction and a vectorized implication elimination (down to atomic types $p$), and the rule for typing alternatives (of the same atomic type) -- and all rules are read coinductively, indicated by the $\cosign$ mark. A well-typed such term hence \emph{locally} conforms to intuitionistic implicational logic.
\begin{figure}[tb]
  \[
    \begin{array}{c}
\infer[\!\!\!\cosign]{\seqt\Gamma{\lambda x^A.N}{A\impl B}}{\seqt{\Gamma,x:A}
NB}\qquad\!\!
\infer[\!\!\!\cosign]
 {\seqt{\Gamma}{x\tuple{N_i}_{i\leq k}}{p}}
  {(x:\vec B\impl p)\in\Gamma\quad\!\forall i\leq k,\,\seqt\Gamma{N_i}{B_i}}\qquad\!\!
      \infer[\!\!\!\cosign]{\seqt{\Gamma}{\sum_{i\leq n}E_i}p}{\forall i\leq n,\,\seqt\Gamma{E_i}p}
    \end{array}
  \]
\caption{Coinductive typing rules for simple types in the application scenario}\label{fig:typingforests}
\end{figure}
For illustration, we give a well-typed forest in graphical form -- the much easier example of Church numerals is found in \cref{sec:churchnumerals}.
Let $\Threet\eqdef((0\impl 0)\impl 0)\impl 0$ for an atom $0$.
This is the simplest type of rank (\ie, nesting depth) 3. We define a closed forest
of type $\Threet$ in \cref{fig:threegraph} \cite[Example 16]{EspiritoSantoMatthesPintoCoinductiveApproachAPAL}.
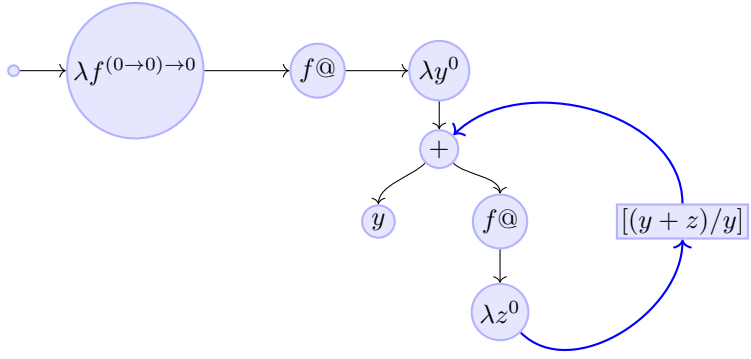
\begin{figure}[tb]
 \begin{tikzpicture}[inner sep=0.5mm,
         place/.style={circle,draw=blue!30,fill=blue!10,thick},
         scale=1.0,
         decont/.style={rectangle,draw=blue!30,fill=blue!10,thick},
         scale=0.8]
         \node[place] at (-7,0) (root)  {};
         \node[place] at (-5,0) (lambdaf)  {$\lambda f^{(0\impl0)\impl0}$};
         \node[place] at (-2,0) (f0)  {$f@$};
         \node[place] at (0,0) (lambday)  {$\lambda y^0$};
         \node[place] at (0,-1.3) (sum)  {$+$};
         \node[place] at (-1,-2.5) (y)  {$y$};
         \node[place] at (1,-2.5) (f1)  {$f@$};
         \node[place] at (1,-4) (lambdaz)  {$\lambda z^0$};
         \node[decont] at (4,-2.5) (decont)  {$[(y+z)/y]$};
         \draw [->] (root) to [out=0,in=180] (lambdaf);
          \draw [->] (lambdaf) to [out=0,in=180] (f0);
          \draw [->] (f0) to [out=0,in=180] (lambday);
          \draw [->] (lambday) to [out=270,in=90] (sum);
          \draw [->] (sum) to [out=225,in=90,looseness=0.8] (y);
          \draw [->] (sum) to [out=315,in=90,looseness=0.8] (f1);
          \draw [->] (f1) to [out=270,in=90] (lambdaz);
          \draw [->, blue, thick] (lambdaz) to [out=315,in=270,looseness=1] (decont);
          \draw [->, blue, thick] (decont) to [out=90,in=45,looseness=1] (sum);
        \end{tikzpicture}
\caption{Forest representation of all inhabitants of $\Threet$}\label{fig:threegraph}
\end{figure}
$f@$ is short for $f\tuple N$ with $N$ given by where the arrow points to.
The ``decontraction operation'' in the back link resides on the meta-level and is specific to the summation in this example grammar:
$[(y+z)/y]$ (written $[y:0,z:0/y:0]$ in the cited paper) is decontraction and says that every occurrence of $y$ has to be replaced by a sum once with $y$ and once with $z$ in place of the original $y$. This forest representation can be seen as a formal approach to the informal concept of ``inhabitation machines'' \cite[pp.\,34--38]{lambdacalculuswithtypes}.
All the inhabitants of $\Threet$ can be read off this forest: omitting types, they are of the form $\lambda f.f\tuple{\lambda y_1.f\tuple{\lambda y_2.f\tuple{\cdots\tuple{\lambda y_n.y_i}\cdots}}}$, with $1\leq i\leq n$.
The individual inhabitants are wellfounded, but the forests representing the entire search spaces (for all simple types) are obtained coinductively in \cite{EspiritoSantoMatthesPintoCoinductiveApproachAPAL}.
We use a generic construction of syntax such as the forests of this scenario that is based on category theory.

\subsection{Context and Overview of this Paper}

For \emph{wellfounded} languages with variable binding, categorical semantics are given in \cite{DBLP:conf/lics/FiorePT99}.
The importance of monoidal structure for the modelling of substitution is emphasized there; many of the constructions are given on the level of monoidal categories, and are later instantiated to a suitable category of contexts.
A very extensive overview of work on substitution for wellfounded syntax with binders, comparing \cite{DBLP:conf/lics/FiorePT99} and subsequent work by the same and other authors, is given by Lamiaux and Ahrens \cite{lamiaux2024introduction}.
A categorical semantics of \emph{non}-wellfounded syntax with binding appeared in \cite{MatthesUustaluTCS} involving the first author. That work is set concretely in endofunctor categories instead of general monoidal categories.

In that context, the present paper makes the following contributions.
Firstly, the new definitions and results of the present paper lift the approach of \cite{MatthesUustaluTCS} to the abstraction level of monoidal categories -- to reach the same abstraction level as \cite{DBLP:conf/lics/FiorePT99}.
Secondly, we provide a full type-theoretic formalization of the results on the abstract level.
Using both of these contributions, by (non-trivial) instantiation, we get a tool chain from multi-sorted binding signatures to certified monadic substitution for non-wellfounded syntax and thus the non-wellfounded counterpart to the tool chain described in \cite{CPP22} involving two of the present authors. Such a tool chain is absent from \cite{MatthesUustaluTCS} even on the informal level (neither multi-sorted binding signatures nor unsorted binding signatures are considered). Our approach is now general enough so that the monoidal category underlying \cite{DBLP:conf/lics/FiorePT99} and its ramifications can also be studied concerning non-wellfounded syntax and substitution for it.

In more detail, we construct non-wellfounded syntax from final coalgebras in suitable functor categories, hence based on category theory.
Variable binding is modelled through the use of nested datatypes, as, e.g., in \cite{birdmeertens,DBLP:journals/jfp/BirdP99}.
The structure map of the final coalgebra is, of course, an isomorphism with its inverse providing an algebra structure.
This excludes the existence of exotic terms.
We benefit from that abstract view in order to construct a monadic substitution operation similar to \cite{DBLP:conf/csl/AltenkirchR99}, \ie, a meta-level operation that is not specific to the application scenario presented above (in contrast to the decontraction operation it features).
The qualifier ``monadic'' implies that the generic approach includes proving the monad laws.
For the case of wellfounded syntax, this is well-established in the literature (see, \eg, \cite{DBLP:conf/csl/AltenkirchR99}).
For the non-wellfounded case but without types, this has also been done before \cite{DBLP:journals/corr/KurzPSV13}.

Further previewing our technical contributions, our approach is to generalize the notion of heterogeneous substitution systems \cite{MatthesUustaluTCS} from endofunctor categories to monoidal categories.
Those systems (abbreviated HSS) were meant as a tool to construct monadic substitution both for wellfounded and non-wellfounded syntax -- by a common abstraction that serves as a pivotal structure between initial algebras and final coalgebras, respectively, on the input side and the substitution monad as output.
We call the generalization \emph{monoidal heterogeneous substitution systems} (MHSS).
All these ingredients have been considered before for the sake of representation of wellfounded syntax \cite[Section I.1.2]{DBLP:conf/lics/Fiore08}\cite[Section 5.2.1]{DBLP:phd/ethos/Hur10}.
\cref{sec:syntaxinmoncat} is the core contribution of this paper, making the step to non-wellfounded syntax on the more abstract level of monoidal categories.
The results in that section demonstrate the pivotal role of MHSS: from a final coalgebra, a MHSS is constructed, and from a MHSS, a monoid is constructed, which abstracts away from monadic substitution.
\cref{sec:msbs} applies the results of \cref{sec:syntaxinmoncat} to the endofunctor scenario -- which is hardwired into the definitions in \cite{MatthesUustaluTCS}.

\subsection{Synopsis}
The remainder of this paper is structured as follows.
In \cref{sec:prelim} we review some prerequisites from category theory that are used later in the paper. \cref{sec:syntaxinmoncat} presents the construction of non-wellfounded syntax with substitution on the level of monoidal categories. As promised above, it generalizes both to monoidal categories and from wellfounded to non-wellfounded syntax -- sloppily construable as ``pushout'' of these two directions.
\cref{sec:msbs} applies the results of \cref{sec:syntaxinmoncat} to the endofunctor scenario, capturing simply-typed non-wellfounded syntax (with binding) generically.
The appendix contains technical complements. \fscdarxiv{However, for lack of space, \cref{sec:appbackground} and \cref{sec:constrptdstrength} are only present in the full version \cite{DBLP:journals/corr/abs-2308-05485}.}{Notice that, for lack of space, \cref{sec:appbackground} and \cref{sec:constrptdstrength} are not present in the published conference version \cite{FSCD24MatthesWullaertAhrens}.}

All of the definitions and results presented in this paper (except for the motivational application scenario in \cref{sec:appscenario}) are formalized and computer-checked in \UniMath~\cite{UniMath}, a library of univalent mathematics based on the computer proof assistant \Coq~\cite{coq}.
Throughout this paper, definitions and results are annotated by \Coq identifiers for the corresponding definitions and results in our library.
These identifiers are hyperlinks leading to an HTML version of the proof code; for instance, clicking on \coqident{CategoryTheory.Monoidal.Categories}{monoidal} brings you to the definition of monoidal category.
The formalization is not the main topic of this paper (a discussion of some formalization aspects is found in \cref{sec:formal}); we use it mainly to relieve ourselves from the burden of writing out lengthy and uninteresting proofs, and the reader from the burden of reading them.
Instead, we restrict ourselves to pointing the reader to interesting aspects of proofs and constructions, and aim to convey the intuition behind -- and useful applications of -- our work.

\section{Preliminaries}\label{sec:prelim}

We assume working knowledge of category theory and mostly only point to specific choices of notation.
We write \( a : \CC\) to indicate that \(a\) is an object in category \(\CC\);
we write %
\(f : a \to b\) to indicate that \(f\) is a morphism from \(a\) to \(b\) in \(\CC\).
Following the choice adopted for the \UniMath library, composition is written in ``diagrammatic'' order, \ie, the composite of \(f : a \to b\) and \(g : b \to c\) is denoted \(\compose{f}{g} : a \to c\). %
We will make use of a category $\set$ of sets; objects of this category are called ``small'' sets.

\subsection{Monoidal Categories and Actegories}\label{sec:moncatsacts}

In this section we briefly review the notions of monoidal category and of actegory.

A monoidal category is given by a six-tuple $(\CC,\otimes,I,\lambda,\rho,\alpha)$ where $\CC$ is a category,
$\otimes:\CC\times\CC\to\CC$, $I:\CC$, $\lambda=(\lambda_x)_{x:\CC}$ (the \emph{left unitor}) with $\lambda_x:I\otimes x\to x$, $\rho=(\rho_x)_{x:\CC}$ (the \emph{right unitor}) with $\rho_x:x\otimes I\to x$ and $\alpha=(\alpha_{x,y,z})_{x,y,z:\CC}$ (the \emph{associator}) with $\alpha_{x,y,z}:(x\otimes y)\otimes z\to x\otimes(y\otimes z)$. %
The unitors and the associator are required to be natural isomorphisms, and are furthermore subject to coherence laws called the ``triangle law'' and the ``pentagon law'', recalled in \fscdarxiv{\cref{sec:appbackground} that is only available in the full version \cite{DBLP:journals/corr/abs-2308-05485}}{\cref{fig:moncatlaws} in the appendix}.
We will use the letter $\VV$ to indicate the first component of a monoidal category and, by slight abuse of language, we even call $\VV$ a monoidal category when the other components are left implicit. We can also just mention $(\CC,\otimes,I)$ or $(\VV,\otimes,I)$.

For the proper understanding of the strength notion and for the construction process of a strength in our application scenario, we use actions of monoidal categories on categories, called \emph{actegories}. For the naming of concepts, we vaguely follow \cite{Actegories}. %

Given a monoidal category $\VV$, %
a (left) $\VV$-actegory is given by a quadruple $(\CC,\odot,\lambda,\actor)$ where $\CC$ is a category,
$\odot:\VV\times\CC\to\CC$ (the \emph{action}), $\lambda=(\lambda_x)_{x:\CC}$ (the \emph{unitor}) with $\lambda_x:I\odot x\to x$, and $\actor=(\actor_{v,w,x})_{v,w:\VV,x:\CC}$ (the \emph{actor}) with $\actor_{v,w,x}:(v\otimes w)\odot x\to v\odot(w\odot x)$.
The unitor and the actor are required to be natural isomorphisms, and are furthermore required to satisfy coherence laws analogous to the ones for monoidal categories, called also ``triangle law'' and ``pentagon law'' that are found in \fscdarxiv{\cref{sec:appbackground} (only in \cite{DBLP:journals/corr/abs-2308-05485})}{\cref{fig:actegorylaws} in the appendix}.
We consider it as important that actegories are a kind of widening of the concept of monoidal categories, in the following sense: Given a monoidal category $(\VV,\otimes,I,\lambda,\rho,\alpha)$, $(\VV,\otimes,\lambda,\alpha)$ is a $\VV$-actegory, and it is called the actegory \emph{with the canonical self-action} (\cf \coqident{CategoryTheory.Actegories.ConstructionOfActegories}{actegory_with_canonical_self_action}).%

We furthermore consider \emph{strong monoidal functors} from monoidal category $(\CC,\otimes,I)$ to monoidal category $(\DD,\otimes',I')$. Such a functor is given by a triple $(F,\epsilon,\mu)$, where $F:\CC\to\DD$ (the \emph{underlying functor}), $\epsilon:I'\to FI$ (preservation of unit) and $\mu=(\mu_{x,y})_{x,y:\CC}$ (preservation of tensor) with $ \mu_{x,y}:Fx\otimes'Fy\to F(x\otimes y)$.
Here, we assume $\mu$ to be a natural transformation, and $\epsilon$ and $\mu$ to be isomorphisms (so as to be ``strong''), as well as the three well-known (lax) laws of preservation of left and right unitality and associativity. %
By abuse of notation, we even call $F$ a strong monoidal functor when the other components are left implicit.

For the purposes of syntax representation, we only consider the lax form of morphisms between actegories, over a common monoidal category $\VV$.
Given a monoidal category $(\VV,\otimes,I)$, a lax linear functor from actegory $(\CC,\odot,\lambda,\actor)$ to actegory $(\DD,\odot',\lambda',\actor')$ is given by a pair $(F,\ell)$, where $F:\CC\to\DD$ (the \emph{underlying functor}) and $\ell=(\ell_{v,x})_{v:\VV,x:\CC}$ (the \emph{lineator}) with $\ell_{v,x}:v\odot' Fx\to F(v\odot x)$.
We require the lineator to be a natural transformation (not necessarily an isomorphism), and furthermore require it to satisfy two laws of preservation of the unitor and the actor, see \cref{fig:linearity}. Currying away the second index to $\ell$ and using that $\lambda$ and $\actor$ are isomorphisms, these laws uniquely determine $\ell_I$ and give a formula to calculate $\ell_{v\otimes w}$ from $\ell_v$ and $\ell_w$. (On this level of generality of the description, this is not different from the situation for the $\mu$ component of a strong monoidal functor.)

\begin{figure}[tb]
  \[
    \begin{tikzcd}
      I\odot' Fx \ar[rr,"\ell_{I,x}"] \ar[dr, "\lambda'_{Fx}"']& & F(I \odot x) \ar[dl,"F\lambda_x"] \\
      & Fx
    \end{tikzcd}\hspace{-3pt}
    \begin{tikzcd}[column sep=small]
      (v\otimes w)\odot'Fx \ar[rr, "\ell_{v\otimes w,x}"] \ar[d, "\actor'_{v,w,Fx}"] & & F((v\otimes w)\odot x) \ar[d, "F\actor_{v,w,x}"']\\
        v\odot'(w\odot' Fx)\ar[dr, "1_v\odot'\ell_{w,x}"'] & &F(v\odot(w\odot x)) \\
        &v\odot'F(w\odot x) \ar[ur, "\ell_{v,w\odot x}"']
      \end{tikzcd}
    \]

\caption{Preservation of the unitor and the actor in a linear functor}\label{fig:linearity}
\end{figure}

\subsection{Pointed Strength}\label{sec:ptdstrength}

Pointed strength is best understood through actegories; this is sketched in \cite[Section I.1.2]{DBLP:conf/lics/Fiore08}), and our presentation here has the same main ingredients (using reindexing and a coslice category, see below). This abstract view is helpful for creating libraries of functors with pointed strength, as will be visible in \cref{sec:msbs}.
Hur uses the notion of pointed strength extensively but only spells it out concretely \cite[Section 5.2.1]{DBLP:phd/ethos/Hur10}.

Given monoidal categories $\WW$ and $\VV$, a strong monoidal functor $F:\WW\to\VV$ and a $\VV$-actegory $(\CC,\odot,\lambda,\actor)$, one can canonically construct a $\WW$-actegory $(\CC,\odot',\lambda',\actor')$ over the same base category -- the \emph{reindexing of the $\VV$-actegory along $F$}. (It seems that it would suffice that $F$ is an oplax monoidal functor instead of a strong one.) On objects, the action $\odot'$ is constructed as $w\odot' x \eqdef Fw \odot x$. We will not spell out the details here; our formalization of this construction is given in \coqident{CategoryTheory.Actegories.ConstructionOfActegories}{reindexed_actegory}.

We need reindexed actegories for one specific situation: the actegory with the canonical pointed action. We assume a monoidal category $\VV$ with unit $I$ and construct the monoidal category of ``monoidal-pointed objects'': the underlying category is the coslice category $I/\VV$ whose objects are pairs $(v,\pointforv)$ with $v:\VV$ and $\pointforv:I\to v$ (``a point for $v$''), and the monoidal category can be easily constructed. %
Just for the record: $I^\ptd\eqdef(I,1_I)$ is the unit, and the tensor is defined on objects as $(v,\pointforv)\otimes^\ptd(w,\pointforw)\eqdef(v\otimes w,\lambda^{-1}_I\cdot(\pointforv\otimes\pointforw))$.

Given a monoidal category $\VV$, the actegory $\VV^\ptd$ that we call the \emph{actegory with the canonical pointed action of $\VV$} is obtained by reindexing: in the definition above, we take $\WW:=I/\VV$, $\VV$ as given, $F$ the forgetful functor that forgets the points (and is strong monoidal), and as $\VV$-actegory the actegory with the canonical self-action of $\VV$ introduced above. It follows that in $\VV^\ptd$, the monoidal-pointed objects of $\VV$ act on the objects of $\VV$.

Given a monoidal category $\VV$ and an endofunctor $F$ on $\VV$, a \emph{pointed tensorial strength} for $F$ is a $\theta$ so that $(F,\theta)$ is a lax endomorphism of actegory $\VV^\ptd$. In other words, $\theta$ is the lineator (following the literature, we use $\theta$ for this specific use of lineators) in the situation where source and target action are $\VV^\ptd$.

In order to allow for an easy comparison with the literature, we spell out the lineator laws for pointed tensorial strength $\theta$ (besides the requirement of naturality in both arguments): the components are $\theta_{(v,\pointforv),x}:v\otimes Fx\to F(v\otimes x)$, and the preservation rules are given in \cref{fig:lawspointedtensorialstrength}. The differences with \cite[Section 5.2.1]{DBLP:phd/ethos/Hur10} are all of presentational nature, most notably that we use left actegories while Hur has the monoidal-pointed objects as second parameter of his ``pointed strength'' $\mathsf{st}$.

\begin{figure}[tb]
  \[
    \begin{tikzcd}
      I\otimes Fv \ar[rr, "\theta_{I^\ptd,v}"] \ar[dr, "\lambda_{Fv}"'] & & F(I \otimes v) \ar[dl, "F\lambda_v"] \\
      & Fv
    \end{tikzcd}
      \begin{tikzcd}[column sep=small]
        (v\otimes w)\otimes Fx \ar[rr, "\theta_{(v,\pointforv)\otimes^\ptd (w,\pointforw),x}"] \ar[d, "\alpha_{v,w,Fx}"] & & F((v\otimes w)\otimes x) \ar[d,  "F\alpha_{v,w,x}"'] \\
        v\otimes(w\otimes Fx)\ar[dr, "1_v\otimes\theta_{(w,\pointforw),x}"'] & &F(v\otimes(w\otimes x)) \\
        &v\otimes F(w\otimes x) \ar[ur, "\theta_{(v,\pointforv),w\otimes x}"']
    \end{tikzcd}
  \]

\caption{Preservation of the unitor and the actor for pointed tensorial strength}\label{fig:lawspointedtensorialstrength}
\end{figure}

\section{Monoid Structure on Non-Wellfounded Syntax}\label{sec:syntaxinmoncat}

In this section, we construct a well-behaved substitution operation on non-wellfounded syntax.
We do so on the level of monoidal categories, using the new notion of ``monoidal heterogeneous substitution system''.
Already mentioned in the introduction, this notion will have a pivotal role in this section: as an intermediate step between a given final coalgebra and the monoid representing substitution on that final coalgebra.
The carrier of these structures is one object $t$ of the given monoidal category $\VV$, so $t$ is the representation of all terms as a whole, thus abstracting away from context/scope and typing details. (In \cref{sec:msbs}, $\VV$ will be instantiated to an endofunctor category, so that such a $t$ will be a functor whose argument is interpreted as a typing context.)

\begin{definition}[Monoidal heterogeneous substitution system, \coqident{SubstitutionSystems.GeneralizedSubstitutionSystems}{mhss}]
  \label{def:mhss}
 Let $\VV$ be a monoidal category with unit $I$, tensor $\otimes$ and right unitor $\rho$ and $H$ an endofunctor on $\VV$ with a pointed tensorial strength $\theta$ for $H$.
 We consider triples $(t,\eta,\tau)$ with $t:\VV$ (the ``terms''), $\eta:I\to t$ (representing the injection of variables into terms) and $\tau:Ht\to t$ (the $H$-algebra representing the domain-specific constructors). Hence $(t,\eta)$ is a monoidal-pointed object. $(t,\eta,\tau)$ is a \emph{monoidal heterogeneous substitution system} (MHSS) for $(\VV,H,\theta)$ if, for all $(z,e,f)$ with
  $z:\VV$, $e:I\to z$ and $f:z\to t$, there is a unique morphism $h:z\otimes t\to t$ such that the following diagram commutes:
  \begin{equation}\label{eq:mhss}
    \begin{tikzcd}[column sep=large, row sep=scriptsize]
      z\otimes I\ar[r, "1_z\otimes\eta"] \ar[dd, "\rho_z"']
      & z\otimes t \ar[dd, "h"]
      & z \otimes H t   \ar[l, "1_z\otimes\tau"'] \ar[d, "\theta_{(z,e),t}"]
      \\
      &
      & H (z\otimes t) \ar[d, "H h"]
      \\
      z \ar[r, "f"]& t
      & H t  \ar[l, "\tau"']
    \end{tikzcd}
  \end{equation}
The uniquely existing morphism $h$ is denoted as $\gst{(z,e)}{f}$.
\end{definition}
Notice that for the considered triples $(z,e,f)$, $(z,e)$ is a monoidal-pointed object. The morphism $f$ is just a $\VV$-morphism and not a ``monoidal-pointed'' morphism from $(z,e)$ to $(t,\eta)$, see~\cref{rem:hss-to-mhss}.
The left unitor and the associator of the monoidal category do not enter this definition directly but through the laws governing $\theta$ (\cf \cref{fig:lawspointedtensorialstrength}).

As seen on the right-hand side of \cref{eq:mhss}, the strength \(\theta\) is an operation that serves to prepare the arguments that are fed into the ``structurally recursive call'' \(Hh\), before applying the domain-specific constructors bundled in $\tau$.
In other words, $h$ mostly follows a homomorphic pattern, except for the rearrangement required by variable binding -- implicitly expressed in functor $H$ -- that is taken care of by $\theta$.

This notion of MHSS is not a recursion scheme specifically for the carrier of an initial algebra (for the functor $I+H-$).
The present notion of MHSS only formulates the ``desideratum'', not sufficient conditions for its fulfillment.
What makes MHSS suitable for dealing with coinductive syntax as well
(when $t,\eta,\tau$ come from a final coalgebra) is the deliberate restriction of the target type of $h$ to $t$.
(This is already part of the notion of heterogeneous substitution system (HSS) \cite{MatthesUustaluTCS}, see~\cref{rem:hss-to-mhss}).
There is also the restriction to $\tau$ as the $H$-algebra in the arrow on the bottom of the diagram, hence
  the limitation to a notion of substitution and not some general recursive pattern.

\begin{remark}\label{rem:hss-to-mhss}
\emph{Monoidal} heterogeneous substitution systems ``almost'' generalize the notion of HSS of \cite{MatthesUustaluTCS} from the specific situation where an endofunctor category $[\CC,\CC]$ is considered to the (unrestricted) monoidal category $\VV$.
The ingredients and stipulations of \cite[Definition 5]{MatthesUustaluTCS} are an instance of our notion of MHSS as soon as $\CC$ has binary coproducts -- so as to be able to speak about an $(\mathsf{Id}+H{-})$-algebra -- modulo the following:
\begin{enumerate}%
\item The order of the arguments of strength $\theta$ is inverted. %
\item We consider all $\VV$-morphisms $f$ and not only morphisms $f$ between the monoidal-pointed objects $(z,e)$ and $(t,\eta)$, satisfying $\eta=e\cdot f$.
  In our diagram, that is \cref{eq:mhss}, $f$ is just the $\VV$-morphism. However, in the diagram in \cite[Definition 5]{MatthesUustaluTCS}, $f$ is written although $Uf$ is meant, with $U$ the forgetful functor from pointed endofunctors to endofunctors forgetting the points.\footnote{This second difference appears to be a
    \emph{conceptual} simplification, and has been formalized in 2022 for endofunctor categories as a ``simplified notion of HSS'', serving as a test bed for our definition of MHSS, \cf \coqfile{}{SubstitutionSystems.SimplifiedHSS.SubstitutionSystems}.}
\end{enumerate}
\end{remark}

For the representation of substitution for non-wellfounded syntax, we will need abstract counterparts to the construction of a monad out of a HSS \cite[Theorem 10]{MatthesUustaluTCS} and the construction of a HSS from a final coalgebra \cite[Theorem 17]{MatthesUustaluTCS}.
It is fair to say that these results carry over to MHSS without difficulty.
We sketch the counterpart to the former construction in \cref{sec:mhsstomonoid} and detail a different path to obtaining the latter in \cref{sec:finalcoalgtomhss}.

\subsection{Construction of a Monoid From a MHSS}\label{sec:mhsstomonoid}

Let $(\VV,\otimes,I,\lambda,\rho,\alpha)$ be a monoidal category.
A $\VV$-monoid is given by a triple $(v,\eta,\mu)$ where $v:\VV$, $\eta:I\to v$ (the ``unit'' of the monoid) and $\mu:v\otimes v\to v$ (monoid ``multiplication''), such that the left and right unit laws and the associative law hold.
We recall the laws in \fscdarxiv{\cref{sec:appbackground} (only in \cite{DBLP:journals/corr/abs-2308-05485})}{the appendix in \cref{fig:monoid}}. %

Let furthermore $H$ be an endofunctor on $\VV$ with a pointed tensorial strength $\theta$ for $H$. An $(H,\theta)$-monoid \cite{DBLP:conf/lics/FiorePT99}
is a quadruple $(v,\eta,\mu,\tau)$ with $(v,\eta,\mu)$ a $\VV$-monoid and $\tau:Hv\to v$ (thus $(v,\tau)$ is an $H$-algebra), such that the following diagram commutes:
\begin{equation}\label{eq:hmonoid}
  \begin{tikzcd}[column sep=3em]
    v\otimes Hv\ar[r,"{\theta_{(v,\eta),v}}"] \ar[d,"1_v\otimes \tau"'] & H(v\otimes v)\ar[r,"H\mu"]& Hv\ar[d,"\tau"]\\
    v\otimes v\ar[rr,"\mu"] && v
  \end{tikzcd}
\end{equation}
The condition expressed by the diagram is the starting point for the parameterization process that ends in the definition of MHSS.

\begin{theorem}[Construction of a monoid from MHSS, \coqident{SubstitutionSystems.GeneralizedSubstitutionSystems}{mhss_monoid}, \coqident{SubstitutionSystems.SigmaMonoids}{mhss_to_sigma_monoid}]\label{thm:sigmamonoidfrommhss}
  We assume the parameters $\VV$, $H$ and $\theta$ of a MHSS. Let $(t,\eta,\tau)$ be a MHSS, and let $\mu:=\gst{(t,\eta)}{1_t}$ the uniquely existing morphism for $(t,\eta,1_t)$. Then $(t,\eta,\mu,\tau)$ is an $(H,\theta)$-monoid.
\end{theorem}

This generalizes \cite[Proposition 3.5]{DBLP:conf/lics/FiorePT99} from the construction of an initial $(H,\theta)$-monoid (under extra sufficient conditions) to the construction of an $(H,\theta)$-monoid from a suitable $H$-algebra (without further conditions).
On the other hand, it lifts \cite[Theorem 10]{MatthesUustaluTCS} from the abstraction level of endofunctor categories to that of monoidal categories.
For the proof of \cref{thm:sigmamonoidfrommhss}, we can precisely follow the organization of the proof of \cite[Theorem 10]{MatthesUustaluTCS}. The absence of the pointedness requirement for $f$ in the definition of MHSS gives rise to an inessential simplification.
The defining diagram of an $(H,\theta)$-monoid is just the right-hand side of \cref{eq:mhss} for the instance used to define $\mu$.
So, we have to establish the monoid laws, for which we only give an overview (\cf \coqident{SubstitutionSystems.GeneralizedSubstitutionSystems}{mhss_monoid} for the formalization).
We define $\mu^{(0)}\eqdef\eta:I\to t$ and $\mu^{(1)}\eqdef\gst{I^\ptd}{\eta}:I\otimes t\to t$.
The morphism $\lambda_t$ satisfies its defining diagram, hence $\mu^{(1)}=\lambda_t$ by uniqueness.
The right unit law of a monoid is just the left-hand side of the defining diagram of $\mu$.
The left unit law of a monoid asks for $\lambda_v$ to be equal to a morphism $m$; since $\lambda_t=\mu^{(1)}$, it suffices to show that $m$ satisfies the defining diagram of $\mu^{(1)}$.
Now, define $\mu^{(2)}\eqdef\mu$. The morphism $\mu^{(2)}$ is even a morphism of the monoidal-pointed objects $(t,\eta)\tensor^\ptd(t,\eta)$ and $(t,\eta)$; the proof uses that $\lambda_I=\rho_I$, which holds generally.
Define $\mu^{(3)}\eqdef\gst{(t,\eta)\tensor^\ptd(t,\eta)}{\mu^{(2)}}:(t\otimes t)\otimes t\to t$.
The associative law of a monoid can now be dealt with by showing that both sides of that equation satisfy the defining diagram of $\mu^{(3)}$ and are hence equal by uniqueness.
The reasoning in both cases is just the monoidal generalization of the first two items of \cite[p.~168]{MatthesUustaluTCS}.

\subsection{Construction of a MHSS From a Final Coalgebra}\label{sec:finalcoalgtomhss}

In this section, we assume the parameters $\VV$, $H$ and $\theta$ of a MHSS. We require binary coproducts in the underlying category of $\VV$ (and use $\tinl$ and $\tinr$ without indices for the left and right injection into the coproduct).
We also assume a final coalgebra $(t,\tout)$ of the functor $(I+H{-})$, \ie, $t:\VV$ and $\tout:t\to I+Ht$.
By Lambek's theorem, $\tout$ is an isomorphism, with inverse $\tout^{-1} : I+Ht \to t$ that can be written as $\tout^{-1} = [\eta,\tau]$ with $\eta:I\to t$ and $\tau:Ht\to t$.
We also require that binary coproducts distribute over the tensor of $\VV$ in its second argument; this means that, for all $v$, $w_1$ and $w_2$, the morphism $[1_v\otimes\tinl,1_v\otimes\tinr]$ from $v\otimes w_1+v\otimes w_2$ to $v\otimes(w_1+w_2)$ has an inverse.
We call that inverse $\delta$ for ``distributor'', without specifying its arguments. %
\begin{theorem}[Construction of MHSS from final coalgebra, \coqident{SubstitutionSystems.ConstructionOfGHSS}{final_coalg_to_mhss_alt}]\label{thm:mhssfromfinal}
  The triple $(t,\eta,\tau)$ is a MHSS for $(\VV,H,\theta)$.
\end{theorem}
For clarity, we deviate from the proof of \cite[Theorem 17]{MatthesUustaluTCS} for HSS, which uses primitive corecursion.
We instead use that $\tout^{-1}$ is a completely iterative algebra (abbreviated as ``cia''), which follows from $\tout$ being a final coalgebra \cite{DBLP:journals/iandc/Milius05}.
In particular, we will only use the definition of cia and not more general corecursion schemes implied by that property.

Given an endofunctor $F$ on a category $\CC$ with binary coproducts, an $F$-algebra $(c,\alpha)$ is called a cia iff for every $x:\CC$ and every morphism $e:x\to Fx+c$ (``a flat equation morphism''), there is a unique morphism $h:x\to c$ that is a ``solution'' of $e$ in $c$ in the sense that the following diagramm commutes:
\[
  \begin{tikzcd}[column sep=5em]
 x \ar[r, "e"] \ar[d, "h"'] &Fx+c\ar[d, "Fh+1_c"]\\
 c&Fc +c \ar[l, "{[\alpha,1_c]}"']
\end{tikzcd}
\]
This generalizes the intuition when $\CC$ is $\SET$: the elements of $x$ are the unknowns, and $e$ either requires a structure in $F$ over the unknowns or directly assigns a value in $c$. A morphism $h$ is a solution if, in the first case, applying $h$ inside the structure and then assembling the structure through $\alpha$ yields the value of $h$ again.

To prove \cref{thm:mhssfromfinal}, we apply the cia scheme for $F:=I+H{-}$ and $\alpha:=\tout^{-1}$.
Given a triple $(z,e,f)$, the following are equivalent:
\begin{enumerate}
\item $h:z\otimes t\to t$ satisfies the defining diagram for $\gst{(z,e)}{f}$ 
\item  $h$ satisfies the defining diagram of a solution for the flat equation morphism $\teqm:z\otimes t\to(I+H{-})(z\otimes t)+t$ defined in \cref{fig:eqmformhss}.
\end{enumerate}
\begin{figure}[tb]
\[
  \begin{tikzcd}[column sep=5em, row sep=scriptsize]
    z\otimes t \ar[r, "\teqm"] \ar[d, "1_z\otimes \tout"'] &(I+H{-})(z\otimes t)+t\\
    z\otimes(I+Ht)\ar[d, "\delta"']\\
    z\otimes I+z\otimes Ht\ar[r, "{\rho_z+\theta_{(z,e),t}}"]&z+H(z\otimes t)\ar[uu, "{[f\cdot\tinr,\tinr\cdot\tinl]}"']
\end{tikzcd}
\]
\caption{Definition of $\teqm$ as composition of four morphisms}\label{fig:eqmformhss}
\end{figure}
The details are found in \cref{sec:proofmhssfromfinal}.
In a nutshell, the two defining diagrams can be massaged so that the equivalence can be seen for each path in the diagrams individually.

On a general note, there is a whole arsenal of categorical corecursion schemes.
For MHSS (and hence for the representation of substitution in the section to come), we picked the method of completely iterative algebras.
Working with these tools from category theory is an alternative to intuitive ``guarded'' definitions and reasoning with observation depths. This alternative is suitable for formalization, for which the present paper is further evidence.

\section{Non-Wellfounded Syntax for Multi-Sorted Binding Signatures}\label{sec:msbs}

In this section, we start from the notion of multi-sorted binding signature (reviewed in \cref{sec:msbsdef}).
Exploiting the high-level results of \cref{sec:syntaxinmoncat} and thus showing their usefulness, we are going to construct the non-wellfounded syntax specified by such a signature, together with a well-behaved -- monadic -- substitution operation on the terms of that syntax.

Our work builds upon previous work \cite{CPP22} involving two of the present authors.
There, categorical semantics of languages of \emph{well}founded terms is developed,
and a construction of the syntax generated by a multi-sorted binding signature is given.
In this section, 
we construct \emph{non}-wellfounded syntax based on that very same notion of multi-sorted binding signatures.
Given such a signature, the existence of the generated syntax is guaranteed by $\omega$-continuity of the associated signature functor -- while, for wellfounded terms, \cite{CPP22} establishes $\omega$-cocontinuity to construct the syntax.
For a modular proof of $\omega$-continuity, we decompose the construction of the associated signature functor slightly differently. Extensionally, we arrive at the same functor, but the formalization of that proof is somewhat intricate.
We therefore suggest the new construction of the signature functor as the one to work with also in the wellfounded case.
This makes good sense if one wants to consider the embedding of wellfounded syntax into non-wellfounded syntax for the same multi-sorted binding signature, and it is doable since we also formalized a proof of $\omega$-cocontinuity.
But there is also an advantage on the conceptual side: the building blocks of the signature functor are all endofunctors, unlike previously \cite{CPP22}.
A second difference with the previous work is that, for the strength construction, we systematically refer to results that reside on the abstract level of monoidal categories.
(Aspects of steps \ref{item:mssig}, \ref{item:sig-functor}, and \ref{item:sig-strength} below are also described in \cite[Section~2]{CPP22}; we discuss them here again for the sake of being self-contained.
The items \ref{item:mhss}-\ref{item:monoidmonad} are concretized in \cref{sec:everything}.)

\begin{enumerate}
\item \label{item:mssig} We describe
simply-typed syntax with variable binding (of finitely many sorted
variables in each constructor argument) as a multi-sorted binding
signature, see \cref{sec:msbsdef}.
\item \label{item:sig-functor} Given a  multi-sorted binding signature, we construct a signature functor $H$ (deviating from \cite{CPP22} for technical reasons), see \cref{sec:signaturefunctor}.
\item We prove $\omega$-continuity of $(\Id{}+H{-})$ and construct the coinductive syntax as the inverse of a final
  coalgebra thereof, see \cref{sec:omegacont}.
\item \label{item:sig-strength} We construct a ``lax lineator'' between actions expressing pointed tensorial strength of $H$, see \cref{sec:everything} including a discussion of the case of simply-typed $\lambda$-calculus and references to the appendix for more details on the general case\fscdarxiv{ (only in \cite{DBLP:journals/corr/abs-2308-05485})}{}.
\item \label{item:mhss} We construct a MHSS (\cref{def:mhss}) for $(H,\theta)$ by applying \cref{thm:mhssfromfinal}.
\item We construct an $(H,\theta)$-monoid by applying \cref{thm:sigmamonoidfrommhss}.
\item \label{item:monoidmonad} Finally we interpret the obtained monoid as monad (hence as monadic substitution) since, during the entire section, we are instantiating the monoidal category to the endofunctors.
\end{enumerate}

\subsection{Multi-Sorted Binding Signatures: Motivation and Definition}\label{sec:msbsdef}

We want to construct syntax of non-wellfounded terms that feature variable binding and have a simple notion of typing. %
Such type systems can be specified using ``multi-sorted binding signatures''; this notion was used, in particular, in \cite{CPP22}, but appears in almost any literature about initial semantics for multi-sorted syntax.
The prime example is simply-typed $\lambda$-calculus (STLC), whose extension to non-wellfounded well-typed terms is an instance of our construction.
We study this example in some detail in \cref{ex:stlc} before reviewing multi-sorted binding signatures in \cref{def:mssigs}.

\begin{example}[Non-wellfounded simply-typed lambda-calculus]\label{ex:stlc}
We are now rephrasing \cite[Example 2.2 and Example 2.10]{CPP22}.
We assume the types of simply-typed $\lambda$-calculus to form a small set $\sort$ that is closed under a binary operation $\Rightarrow : \sort \to \sort \to \sort$. The elements of $\sort$ are called \emph{sorts}, so as to distinguish them from the types of our ambient type theory.
We model syntax over a base category $\CC$ (with initial object $\bot$, terminal object $\top$, and binary products and binary coproducts), not necessarily the category $\set$; however, we motivate the notions for the special case where $\CC$ is $\set$.
Let $\ccsort$ be the functor category $[\sort,\CC]$ where $\sort$ is viewed as a discrete category. In the case when $\CC$ is $\set$, objects of this category are simply functions $\xi : \sort \to \set$, and we generally use letter $\xi$ to indicate objects of $\ccsort$. They represent the typing contexts since $\xi s$ represents the totality of variables of sort $s$.

For the instantiation of \cref{sec:syntaxinmoncat}, we take $\VV$ as the monoidal category of endofunctors of $\ccsort$ -- with the tensor operation $X\otimes Y\eqdef X\cdot Y$ in diagrammatic order. In \cref{def:mhss}, we are looking for one object $T$ of $\VV$ (\ie, $T:[\ccsort,\ccsort]$) as representation of all the wellfounded and non-wellfounded terms, here of simply-typed $\lambda$-calculus. On objects, $T$ assigns to $\xi : \ccsort$ and $s : \sort$ the object $T~\xi~s$ of $\CC$, a representation of all the wellfounded and non-wellfounded terms that have sort $s$ in the typing context $\xi$. The functor $H$ of \cref{sec:syntaxinmoncat} prepares for the construction of $T$ as a fixed point. Instead of only considering the ``solution'' $T$ as an argument to $H$, we have to abstract over an arbitrary $X:\VV$ as an object argument to $H$.
We would like to take $H$ as the pointwise coproduct of one summand for application and one for abstraction, for each pair $(s,t)$ of sorts that parameterize the respective typing rules, \ie $H\eqdef \sum_{s,t:\sort}(\app_{s,t}+\lam_{s,t})$.
Here the summands have to be endofunctors on $[\ccsort,\ccsort]$, and we only give the definition for objects (in all arguments) -- where the defining equation is between objects of $\CC$:
\begin{align*}
    \app_{s,t}~X~\xi~u &\eqdef \left\{
                         \begin{array}{l@{\quad\!}l}
                           X~\xi~(s \Rightarrow t) \times X~\xi~s&\mbox{if $u=t$}\\
                            \bot &\mbox{else}\\
                         \end{array}\right.
                         \enspace, \\
  \lam_{s,t}~X~\xi~u &\eqdef \left\{
                         \begin{array}{l@{\quad\!}l}
                           X~\xi'~t&\mbox{if $u=(s \Rightarrow t)$, with $\xi's\eqdef\top+\xi s$ and $\xi' u'=\xi u'$ for $u'\neq s$}\\
                            \bot &\mbox{else}\\
                         \end{array}\right.\\[-3ex]
\end{align*}
These summands represent all well-sorted applications and $\lambda$-abstractions made from the material in the yet arbitrary object $X$ of $\VV$,
while the variables are dealt with separately from $H$ through the unit of $\VV$ (in our case the identity functor) in \cref{sec:syntaxinmoncat}. 
 The case distinctions on equality of sorts in this motivating example can be avoided following \cite{CPP22}, and we will do so in the generic construction in \cref{sec:signaturefunctor}.
 \end{example}

 We fix a small set $\sort$ representing the sorts.

\begin{definition}[\protect{\cite[Definition 2.1]{CPP22}}, \coqident{SubstitutionSystems.MultiSortedBindingSig}{MultiSortedSig}]\label{def:mssigs}
  A \emph{multi-sorted binding signature} is given by a small set $I$ together with an arity function $\arity : I \to (\sort^* \times \sort)^* \times \sort$.
\end{definition}
Here, we write $A^*$ for the set of finite lists formed from elements of $A$.
The intuition is as follows: for any $i : I$, $\arity(i)$ is the signature of a term constructor.
The second component of $\arity(i)$ is the sort of the constructed term.
The first component is a list of signatures of the arguments of that constructor.
Each such signature is an element of $\sort^*\times\sort$, describing the sorts of all the (anonymous) variables bound by that argument, together with the sort of the argument itself.
\cite{CPP22} makes no claim on originality of that definition, see the discussion there.
It should be stressed that, while $\sort$ and $I$ can be infinite sets, each term constructor described by an $\arity(i)$ only has finitely many arguments.
Non-wellfounded syntax with these constructors is therefore still finitary in the sense that terms, when viewed as trees, are finitely branching.

\begin{example}[\protect{\cite[Example 2.2]{CPP22}}, \coqident{SubstitutionSystems.STLC_alt}{STLC_Sig}]\label{example:stlc}
  Assume that $\sort$ is closed under a binary operation $\Rightarrow$.
  We put into $I$ the sort parameters of the typing rules of the term constructors of STLC.
  Thus, $I$ is taken to be $(\sort \times \sort) + (\sort \times \sort)$.
  The left summand pertains to the application operation while the right summand describes $\lambda$-abstraction:
   \[\arity(\tinl\langle s,t\rangle)\eqdef \big\langle[\langle[],s\Rightarrow t\rangle,\langle[],s\rangle],t\big\rangle\qquad\qquad
    \arity(\tinr\langle s,t\rangle)\eqdef \big\langle[\langle[s],t\rangle],s\Rightarrow t\big\rangle\]
\end{example}

\begin{example}[\coqident{SubstitutionSystems.UntypedForests}{UntypedForest_Sig}]\label{example:untypedforests}
  We model the grammar of the untyped version of the forests described in \cref{sec:appscenario} as a multi-sorted binding signature.
  Let $\sort_0:=\{\sortv,\sortt,\sorte\}$ be a three-element set, having sorts for the three syntactic categories of term variables (\sortv), terms (\sortt) and elimination alternatives (\sorte).
  The first sort seems unavoidable since the elimination alternative $x \tuple{N_1,\ldots,N_k}$ only allows term variables in the head position, not arbitrary terms.
  The index set $I$ represents the different forms of expressions that are parameterized in the elements of the syntactic categories: one for $\lambda$-abstraction, one for each index $n$ for summation, one for each index $k$ for tupling, hence we set $I:=(1+\nat)+\nat$.
We define  \(\arity(\tinl(\tinl *))\eqdef\big\langle[\langle[\sortv],\sortt\rangle],\sortt\big\rangle\), a simplified version of the second case for STLC: a term variable is being bound in a term, yielding a term. The other forms of expressions do not feature variable binding:
\[\arity(\tinl(\tinr\,n))\eqdef\big\langle[\underbrace{\langle[],\sorte\rangle,\ldots,\langle[],\sorte\rangle}_{n}],\sortt\big\rangle\qquad\qquad\arity(\tinr\,k)\eqdef\big\langle[\langle[],\sortv\rangle,\underbrace{\langle[],\sortt\rangle,\ldots,\langle[],\sortt\rangle}_{k}],\sorte\big\rangle\]
None of the arities have $\sortv$ as second component, hence term variables will only come from a given context.
Using the pipeline of \cref{sec:everything}, we can represent these untyped forests as an object $T$ of $[\CC^{\sort_0},\CC^{\sort_0}]$, analogously to \cref{ex:stlc}.
Again analogously to that example, objects $\xi$ of $\CC^{\sort_0}$ represent contexts; in this untyped version they correspond just to a choice of names for all occurring variables.
Then, because of the absence of $\sortv$ as second component of our arities, we should have that $T\xi\sortv$ and $\xi\sortv$ are isomorphic for any $\xi:\CC^{\sort_0}$.
In our representation, we are interested in those $\xi:\CC^{\sort_0}$ for which $\xi\sortt$ and $\xi\sorte$ are the initial object of $\CC$ -- empty sets in the case of $\CC = \SET$ -- so that $\xi$ only provides names for term variables.
For those $\xi$, the untyped forests in the terms category and elimination alternatives category, of the grammar in \cref{sec:appscenario}, are represented as $T\xi\sortt$ and $T\xi\sorte$, respectively.
\end{example}

\begin{example}[\coqident{SubstitutionSystems.Forests}{Forest_Sig}]\label{example:typedforests}
  The typed forests of \cref{sec:appscenario} have as the set of sorts the set $\sort\times\sort_0$, with $\sort$ and $\sort_0$ from the two previous examples.
  Moreover, since atomic types play a specific role in the typing rules in \cref{fig:typingforests}, we have to assume a set $\atom$ of atoms and an operation $\atotype:\atom\to\sort$ (which should be thought of as an inclusion).
  The typing rule for tuples is additionally (as compared to the raw syntax) parameterized by a list of $k$ elements of $\sort$ and one element of $\atom$.
  Accordingly, the index set for this multi-sorted binding signature is
  $I:=(\sort\times\sort + \atom\times\nat)+\sort^*\times\atom$, using the set $\sort^*$ of finite $\sort$-lists introduced above. We define \(\arity(\tinl(\tinl\langle s,t\rangle))\eqdef\big\langle[\langle[\langle s,\sortv\rangle],\langle t,\sortt\rangle\rangle],\langle s\Rightarrow t,\sortt\rangle\big\rangle\), which combines the second case of \cref{example:stlc} and the first case of \cref{example:untypedforests}. The other definitions are as follows:
  \begin{align*}
    \arity(\tinl(\tinr\,\langle p,n\rangle))&\eqdef \big\langle[\underbrace{\langle[],\langle \atotype\,p,\sorte\rangle\rangle,\ldots,\langle[],\langle \atotype\,p,\sorte\rangle\rangle}_{n}],\langle \atotype\,p,\sortt\rangle\big\rangle\\
    \arity(\tinr\,\langle[B_1,\ldots,B_k],p\rangle)&\eqdef\big\langle[\langle[],\langle B,\sortv\rangle\rangle,\langle[],\langle B_1,\sortt\rangle\rangle,\ldots,\langle[],\langle B_k,\sortt\rangle\rangle],\langle\atotype\,p,\sorte\rangle\big\rangle\enspace,
  \end{align*}
  with $B:=B_1\Rightarrow\ldots\Rightarrow B_k\Rightarrow \atotype\,p$, parenthesized to the right.
\end{example}

A multi-sorted binding signature is just simple syntactic data (of a signature), so rather the description of a task to define the intended syntax -- which in our case will include non-wellfounded terms.
In the next section we discuss how to transform such a signature into a more ``semantic'' kind of signature: a functor $H$ such that the semantics of the signature is given by \((\Id{}+H{-})\)-algebras.
While we did this ``by hand'' for STLC in \cref{ex:stlc} (but limited the description just to object arguments), it should be clear that an automatic generation for more involved grammars such as \cref{example:typedforests} would be desirable.

\subsection{(Modified) Signature Functor for Multi-Sorted Binding Signatures}\label{sec:signaturefunctor}

Here, we associate to any multi-sorted binding signature $(I,\arity)$ a suitable functor \(H\), such that the non-wellfounded syntax generated by $(I,\arity)$ is, in particular, an \((\Id{}+H{-})\)-(co)algebra.
\Cref{def:sorted_option_functor,def:option-list,def:proj,def:hat} provide the building blocks for building $H$ modularly from basic constructions, following the combinatorial structure of $(I,\arity)$ as a family of pairs containing lists.

We now use $\ccsort$ generally for the functor category $[\sort,\CC]$, not just for the STLC example in \cref{sec:msbsdef}.
One can still think of $\CC$ as being $\set$, but we keep the category $\CC$ abstract and collect requirements on $\CC$ on the way -- that are all fulfilled by $\set$.
As mentioned in \cref{sec:msbsdef}, we instantiate $\VV$ of \cref{sec:syntaxinmoncat} with the endofunctors on $\ccsort$ and now have to determine the endofunctor $H$ on $[\ccsort,\ccsort]$.

In order to do this, we assume that $\CC$ has a terminal object $\top$, binary products, and set-indexed coproducts (including initial object $\bot$ and binary coproducts).

\begin{definition}[\protect{\cite[Definition 2.3]{CPP22}}, \coqident{SubstitutionSystems.MultiSorted_alt}{sorted_option_functor}]\label{def:sorted_option_functor}
  Let $s$ be a sort. The \emph{sorted option functor}
  $\option_s : \ccsort \to \ccsort$ is defined (on objects) as
   \( \option_s~\xi~t \eqdef  \coprod_{(s = t)} \top\, + \, \xi~t \enspace\).
\end{definition}
In this definition, we form a coproduct in $\CC$ of $\top$ over the type of proofs that $s = t$; \ie, we form a subsingleton.
We thus avoid the use of case distinction: %
$\option_s~\xi$ is an equivalent replacement for $\xi'$ in the definition of $\lam$ above \cite[Remark 2.4, Remark 2.8]{CPP22}.

\begin{definition}[\protect{\cite[Definition 2.5]{CPP22}}, \coqident{SubstitutionSystems.MultiSorted_alt}{option_list}]\label{def:option-list}
  Given a non-empty list of sorts $\ell \convert [s_1, \ldots, s_n]$, $\option^*~\ell : \ccsort \to \ccsort$ is defined as
    \(\option^*~\ell \eqdef \option_{s_1} \hcomp \, (\option_{s_2} \hcomp \, \ldots)\enspace\).
  For an empty list, it is $\option^*~[] \eqdef \Id{}$.
\end{definition}

\begin{definition}[\protect{\cite[Definition 2.6]{CPP22}}, \coqident{SubstitutionSystems.MultiSorted_alt}{projSortToC}]\label{def:proj}
  For any $s : \sort$ the \emph{projection functor} $\proj{s} : \ccsort \to \CC$ is defined (on objects) as:
    \(\proj{s}~\xi \eqdef \xi~s\enspace\).
\end{definition}

\begin{definition}[\protect{\cite[Definition 2.7]{CPP22}}, \coqident{SubstitutionSystems.MultiSorted_alt}{hat_functor}]\label{def:hat}
  For any $s : \sort$ we have a left adjoint %
  to $\proj{s}$, written $\hat{s} : \CC \to\ccsort$, defined on objects as
    \(\hat{s}~c~t \eqdef \coprod_{(s = t)} c\enspace\).
\end{definition}
Here, we essentially define $\hat{s}~c~s=c$ and $\hat{s}~c~t=\bot$ otherwise.
As above, we avoid the case analysis for $\app_{s,t}~X~\xi~u$ in our STLC example in \cref{sec:msbsdef}, hence do not need the matching of the constructor's target type.

We now have all the basic building blocks to associate, to a given multi-sorted binding signature $(I,\arity)$, a signature functor $H: [\ccsort,\ccsort] \to [\ccsort,\ccsort]$.
We turn to the construction of the corresponding building blocks for $H$, involving a formal argument $X : [\ccsort,\ccsort]$ to $H$.
Intuitively, $X$ is the unknown functor that, after applying our results of \cref{sec:syntaxinmoncat}, will be set to the functor representing the coinductive sorted syntax.
For the sake of motivating the modification of the definition of $H$ compared to \cite{CPP22}, we will now work top-down and use the letter $G$ with upper indices (instead of $F$ with upper indices, used there).

The final step of the construction of $H$ is unchanged from \cite{CPP22}.
Assume that we already have a functor $G^{(\vec a,t)}: [\ccsort,\ccsort] \to [\ccsort,\ccsort]$ for all $(\vec a,t):(\sort^* \times \sort)^* \times \sort$.
Then, for the multi-sorted binding signature $(I,\arity)$, the associated signature functor $H:[\ccsort,\ccsort] \to [\ccsort,\ccsort]$ is given by the (pointwise) coproduct $HX\eqdef \coprod_{i~:~I} G^{\,\arity(i)}X$ (\coqident{SubstitutionSystems.ContinuitySignature.ContinuityOfMultiSortedSigToFunctor}{MultiSortedSigToFunctor'}).

The penultimate step constructs $G^{(\vec a,t)}$ (\coqident{SubstitutionSystems.ContinuitySignature.ContinuityOfMultiSortedSigToFunctor}{hat_exp_functor_list'_optimized}).
It assumes given functors $G^{(a,t)}: [\ccsort,\ccsort] \to [\ccsort,\ccsort]$ for all $(a,t):(\sort^* \times \sort)\times \sort$.
Given a non-empty list $\vec a\convert[a_1, \ldots, a_n]$, $G^{(\vec a,t)}$ is defined (on objects) as the iterated pointwise binary product $ G^{(a,t)}X\eqdef G^{(a_1,t)}X \times(G^{(a_2,t)}X\times\ldots)$.
The corner case $G^{([],t)}$ is given (maybe peculiarly) by $G^{([],t)}X\eqdef \top_{[\ccsort,\CC]} \cdot \hat t$, regardless of $X$, with $\top_\DD$ the terminal object of $\DD$.
(The terminal object in functor categories is given by constantly the terminal object of the target category.)
More precisely, $G^{([],t)}$ is the composition (in diagrammatic order) of first $\top_{[[\ccsort,\ccsort],[\ccsort,\CC]]}$ and second post-composition with $\hat t$ (which is a functor from $[\ccsort,\CC]$ to $[\ccsort,\ccsort]$).
This view will be exploited for the construction of a pointed tensorial strength for $G^{([],t)}$.

It remains to construct $G^{(a,t)}$ (\coqident{SubstitutionSystems.ContinuitySignature.ContinuityOfMultiSortedSigToFunctor}{hat_exp_functor_list'_piece}).
It is a refined version of $F^a$ in \cite{CPP22} since it takes into account the ``target'' sort $t$.
In fact $G^{(a,t)}$, corresponds to $F^a\cdot \hat t$ in terms of that paper, and so the difference between our $G^{(\vec a,t)}$ and $F^{(\vec a,t)}$ in that paper is whether the composition with $\hat t$ is on each component of the product (our solution) or only on the product itself.
That previous solution looks less convoluted but requires the consideration of non-endofunctors.

Let $a\convert(\ell,s)$ with $\ell:\sort^*$.
The object part of functor $G^{(a,t)}$ is defined as $G^{(a,t)}X\eqdef \option^*~\ell  \hcomp X \hcomp \proj{s} \hcomp \hat t$.
More precisely, $G^{(a,t)}$ is the composition (in diagrammatic order) of first the precomposition with $\option^*~\ell$ and second the postcomposition with $\proj{s} \hcomp \hat t$.
This view will be important to establish $\omega$-continuity of $G^{(a,t)}$.

The instance for STLC can be compared to our introductory example in \cref{sec:msbsdef} -- assuming decidability of the set of sorts, they coincide mathematically.

\subsection{Existence of Final Coalgebra for Functor $\Id{}+H{-}$}\label{sec:omegacont}

Given a multi-sorted binding signature, we want to apply \cref{thm:mhssfromfinal} with $\VV$ the endofunctors on $\ccsort$ and the signature functor $H$ defined in the previous section.
Its first requirement is a final coalgebra of the functor $(\Id{}+H{-})$, with $\Id{}$ the identity functor on $\ccsort$ (which is the unit of $\VV$).
We get a final coalgebra through the dual of Adámek's theorem on the existence of initial algebras for $\omega$-cocontinuous functors on $\omega$-cocomplete categories with initial object.
We now require that $\CC$ is $\omega$-complete (\ie, $\CC$ has limits of shape $0\leftarrow 1\leftarrow 2\leftarrow\cdots$), whence $\VV$ is also $\omega$-complete. (We had already generally required a terminal object $\top$ for $\CC$, which gives one for $\VV$.)
We argue first that $H$ is $\omega$-continuous, then that $(\Id{}+H{-})$ is $\omega$-continuous.

We analyze the building blocks of $H$.
$G^{(a,t)}$ is defined as the composition of two functors, so it is $\omega$-continuous if both functors are.
We prove that postcomposition with $\proj{s} \hcomp \hat t$ is $\omega$-continuous (\coqident{SubstitutionSystems.ContinuitySignature.ContinuityOfMultiSortedSigToFunctor}{post_comp_with_pr_and_hat_is_omega_cont}).
For this, we require of $\CC$ that $\omega$-limits distribute over sub-singleton coproducts.
This means that the canonical morphism from the coproduct of the respective limits to the limit of the coproducts is an isomorphism.
Precomposition with any fixed functor is $\omega$-continuous, hence in particular for $\option^*_\ell$.

We move to $\omega$-continuity of $G^{(\vec a,t)}$, which we prove by induction on the length of $\vec a$. For an empty $\vec a$, this is a constant functor; and pointwise binary products preserve $\omega$-continuity.

Moreover, we require of $\CC$ that $\omega$-limits distribute over $I$-coproducts, a property that then also holds of $\VV$.
Then, $\omega$-continuity of the $G^{\,\arity(i)}$ carries over to their coproduct $H$.

As a final step, we need to show $\omega$-continuity of $(\Id{}+H{-})$ (\coqident{SubstitutionSystems.MultiSortedMonadConstruction_coind_actegorical}{is_omega_cont_Id_H}).
In order to avoid still other hypotheses about distribution of $\omega$-limits over certain colimits, we use that binary coproducts are (isomorphic) to $\bool$-indexed coproducts.
So, the final assumption on $\CC$ is that $\omega$-limits distribute also over $\bool$-coproducts (besides over sub-singleton coproducts and $I$-coproducts for the index set $I$ of the multi-sorted binding signature).

For the most interesting case $\CC=\SET$, all these requirements are met.

\subsection{Putting Everything Together}\label{sec:everything}

We can now apply the general results of
\cref{sec:syntaxinmoncat} for the construction of coinductive syntax with monadic substitution, as specified by a
multi-sorted binding signature.
We fix a set $\sort$ of sorts and a multi-sorted binding signature $(I,\arity)$ over $\sort$.
This means in particular that we take as $\VV$ the endofunctors on $\ccsort$ when instantiating the results of \cref{sec:syntaxinmoncat}.
We denote by \(H : [\ccsort,\ccsort] \to [\ccsort,\ccsort]\) the signature functor associated to $(I,\arity)$ according to the construction in \cref{sec:signaturefunctor}.

We aim to apply \cref{thm:mhssfromfinal} to the final \((\Id{}+H{-})\)-coalgebra obtained in \cref{sec:omegacont}.
To this end, the parameter $\theta$ of a MHSS has to be specified, \ie, we have to construct a suitable pointed tensorial strength $\theta$ for the signature functor $H$.
To recall, $\theta$ instructs via the notion of $(H,\theta)$-monoid (\cref{eq:hmonoid}) how the monoid multiplication $\mu$ acts on the constructors (embodied in the $H$-algebra $\tau$).
But in our situation, $\mu$ is the monad multiplication expressing substitution -- it is well-known that $\VV$-monoids in this case are nothing but monads on $\ccsort$ (\cf\coqident{CategoryTheory.Monoidal.Examples.MonadsAsMonoidsElementary}{monoid_to_monad_CAT}).
In other words, we have to define $\theta$ so that it describes correctly the recursive behaviour of substitution (as expressed by $\mu$).

\begin{example}[Strength for STLC]
  A suitable $\theta$ for the ``hand-written'' signature functor $H$ in \cref{ex:stlc} contains, in particular, operations from $Y\cdot (\app_{s,t}X)$ to $\app_{s,t}(Y\cdot X)$ and from $Y\cdot (\lam_{s,t}X)$ to $\lam_{s,t}(Y\cdot X)$ for any $s,t:\sort$, endofunctor $Y$ on $\ccsort$ that has a point $\eta:\Id{}\to Y$, %
  and endofunctor $X$ on $\ccsort$.
The instance relevant in \cref{eq:hmonoid} is then with $X$ and $Y$ the representation $T$ of STLC, and with $\eta$ the inclusion of variables in terms $T$.
The operations for \emph{application} do essentially nothing since $\mu$ should just descend into the subtrees.
For \emph{abstraction}, given a $\xi:\ccsort$, we need to specify a function of type \(T\xi+1_s \to T(\xi+1_s)\), where \(1_s : \ccsort\) is defined as \(1_s(t)\defeq (t = s)\); that is, \(X + 1_s\) is \(X\) extended by an element of sort \(s\).
In other words, we need to lift $T\xi:\ccsort$ extended by an element $\star$ of sort $s$
to $T\xi'$, with $\xi'$ the extension of $\xi$ by an element $\star$ of sort $s$ (as in \cref{ex:stlc}); this will need $\eta$ for \(\star\) in the input.
In short, for abstraction, the strength specifies essentially the famous ``lift'' operation on a substitution function to avoid capture when descending under a binder.
\end{example}

The construction of a suitable $\theta$ for the variant of the generic $H$ considered in \cite{CPP22} could have been adapted to our $H$, but we have preferred to give a construction that, although working on the level of endofunctors, is formed from building blocks that reside on the general level of monoidal categories. To structure this construction, we propose the notion of \emph{relative lax commutator} that generalizes the notion of ``pointed distributive law'' in \cite[Definition 10]{signatures_to_monads}.
\fscdarxiv{For lack of space, t}{T}he whole (technical) construction is explained in \cref{sec:constrptdstrength}\fscdarxiv{, but only in \cite{DBLP:journals/corr/abs-2308-05485}}{}.
We thus take as $\theta$ the pointed tensorial strength for $H$ described in the appendix. For STLC (\cref{example:stlc}), the strength can be exploited on the abstract level with base category $\CC$ (\cf\coqident[thetaSTLC]{SubstitutionSystems.STLC_actegorical_abstractcat}{a4d147d62cd033ae14f5653d183758ed} in the formalization). For forests (\cref{example:untypedforests} and \cref{example:typedforests}), we only exploited the situation with $\CC$ set to $\SET$ (\cf\coqident[thetaUntypedForest]{SubstitutionSystems.UntypedForests}{f9f6d2c8a60aa898e8c41345b02b9c5c} and \coqident[thetaForest]{SubstitutionSystems.Forests}{458425600edbc6e909943333ac02c4ff} in the formalization).

\cref{thm:mhssfromfinal} provides us with a MHSS that serves as input to \cref{thm:sigmamonoidfrommhss}, hence we get an $(H,\theta)$-monoid.
This monoid is, in particular, a monad, our (certified) substitution monad for the non-wellfounded syntax described by the given multi-sorted binding signature. For STLC, we have formalized this on the abstract level and constructed the finite Church numerals as well as the infinite Church numeral (\cf\coqfile{}{SubstitutionSystems.STLC_actegorical_abstractcat}).

To be more concrete, we can instantiate the base category $\CC$ to $\SET$ that satisfies all requirements on $\CC$ we made during the construction process and get a set of well-sorted non-wellfounded terms for any sort, given a supply of sets of variables for any sort, together with a substitution operation that respects sorts and satisfies the monad laws.
We replayed the construction of the finite Church numerals in STLC in this concrete setting (\cf\coqfile{}{SubstitutionSystems.STLC_actegorical}). For forests, we only considered the $\SET$ case and have instantiated the general constructions (\cf\coqfile{}{SubstitutionSystems.UntypedForests} and \coqfile{}{SubstitutionSystems.Forests}).

Although this not the topic of this paper, we mention that we also have adapted the results and the formalization of wellfounded syntax to the present setting:
This includes the construction of a MHSS from an initial $(I+H{-})$-algebra under the proviso it has been obtained through a Mendler-style construction based on $\omega$-cocontinuity of $H$ (\coqident{SubstitutionSystems.ConstructionOfGHSS}{initial_alg_to_mhss}).
We further established that this MHSS gives rise to an initial $(H,\theta)$-monoid \coqident{SubstitutionSystems.ConstructionOfGHSS}{SigmaMonoidFromInitialAlgebraInitial}, for the given strength $\theta$.
These results are on the level of a monoidal category, as in our \cref{sec:syntaxinmoncat}.
The signature functor $H$ constructed in \cref{sec:signaturefunctor} is even $\omega$-cocontinuous (based on the proof of the same property for the variant considered in \cite{CPP22}), under conditions on $\CC$ that are fulfilled for $\SET$.
Thus, $(\Id{}+H{-})$ is $\omega$-bicontinuous, and we get a morphism of $(H,\theta)$-monoids from the inductive to the coinductive syntax (\coqident{SubstitutionSystems.MultiSortedEmbeddingIndCoindHSET}{ind_into_coind} for the case $\SET$).
The above-mentioned four \Coq vernacular example files in our \UniMath library illustrate that, thanks to that actegorical development, the use of the formalized wellfounded and the formalized non-wellfounded syntax for those multi-sorted binding signatures can be done in parallel. For example, the finite Church numerals in STLC are developed independently of the choice for one of these two options. 
This conforms to the intuition that every single wellfounded term belonging to the non-wellfounded syntax already belongs to the wellfounded syntax, even though the categorical development of these structures is very different.

\section{Related Work and Conclusions}\label{sec:concl}

We have cited throughout the paper the work we rely on or which initiated a line of thought. Here, we give additional information on other related work (that may have been also cited already in the main text).
\cite{DBLP:journals/corr/KurzPSV13} also have codatatypes and define datatype-generic substitution corecursively, and they even calculate infinitary normal forms for their example of untyped $\lambda$-calculus.
However, they do not consider typed systems, and the results are not presented on the abstraction level of monoidal categories. Instead, they use a concrete ``nominal'' presentation of syntax with binders.
\cite{DBLP:journals/jfp/AllaisACMM21} also have codatatypes and even datatype-generic programming not only of substitution, but the work is not based on category theory (and so the approach is rather axiomatic than definitional). That work is implemented in the Agda system.
\cite{Zsido-phd} considers different categorical models of simply-typed wellfounded syntax. In its Chapter 5, the monoidal category corresponding to the framework of \cite{DBLP:conf/lics/FiorePT99} is laid out in detail for simple types, and its Chapter 7 compares it with the monoidal category of endofunctors over a slice category. The latter is close to the concrete instance we are studying in \cref{sec:msbs}, but we deal with non-wellfounded syntax. All in all, \cite{Zsido-phd} has a lot on the strength construction with actegories, including for the typed case, and this for more than one concrete categorical representation, but non-wellfounded syntax is not considered.
\cite{DBLP:journals/pacmpl/BlanchetteGPT19} have an approach to codatatypes that is definitionally based on category theory; but it is strongly tied to set theory through infinite cardinal numbers that appear in the definition of the class of ``bounded natural functors'' they consider. This allows them to implement the approach in the Isabelle system (based on a very small kernel).
Popescu \cite{DBLP:journals/pacmpl/Popescu24} compares different corecursors for syntax with variable binding in nominal style; it is partially formalized in Isabelle/HOL.
\cite{DBLP:journals/pacmpl/FioreS22} also translate multi-sorted binding signatures into signatures with strength. Their notion of syntax includes ``meta-variables'', but they stay within the wellfounded terms and heavily use inductive families as provided by the Agda system.
\cite{DBLP:conf/lics/BorthelleHL20} comes with a \UniMath/\Coq formalization of the whole chain even from \emph{skew} monoidal categories to an initial $(H,\theta)$-monoid, hence for wellfounded syntax. Beware that swapping of the arguments of the tensor is not a harmless operation for skew monoidal categories, so our present definition of MHSS does not fit as a pivotal element in their development.
\cite{DBLP:conf/fossacs/HirschowitzHLM22} rework the approach of \cite{DBLP:conf/lics/FiorePT99}, still using pointed strength and $(H,\theta)$-monoids. They also deal with simple types but do not consider non-wellfounded syntax.

We have presented, through the notion of monoidal heterogeneous substitution system, a tool which provides a monadic substitution operation also for non-wellfounded syntax, and this for the first time on the abstraction level of monoidal categories.

Our definitions and results unify the construction of both wellfounded and non-wellfounded syntax with substitution.

We also instantiated monoidal heterogeneous substitution systems to endofunctor categories and adapted the full chain from multi-sorted binding signatures to substitution for non-wellfounded syntax. For the sake of this instantiation, we provide modular results to prove $\omega$-cocontinuity of signature functors and hence obtain both $\omega$-cocontinuity and $\omega$-continuity for the signature functors we generate from multi-sorted binding signatures.

All the results of this paper have been rigorously formalized with \UniMath/\Coq.
For the specific category of sets (types of homotopy level 2 according to univalent foundations) as base category, the hypotheses of the construction of non-wellfounded syntax can be proved.
Hence, for this base category, we have a ``concrete'' formalization of the tool chain, which provides in particular a formal construction in univalent foundations of non-wellfounded syntax with binding, as instructed by a multi-sorted binding signature, and its monadic substitution operation.

A question we have left open is that of equations on non-wellfounded terms, for instance, $\beta$-equivalence.
We anticipate that some definitions could carry over from the wellfounded setting, like the definition of equations and reductions given in \cite{DBLP:conf/rta/AhrensHLM19,DBLP:journals/pacmpl/AhrensHLM20}.
The construction of suitable terminal coalgebras, however, seems to require some work.

\fscdarxiv{}{\bibliography{literatureFSCD24}}

\appendix

\section{An Easy Example for the Application Scenario}\label{sec:churchnumerals}

To ease the understanding of the example given in \cref{sec:intro}, we also show the one of Church numerals.

Let $0$ be a base type. We define a closed forest
of type $(0\to 0)\to 0\to 0$ in \cref{fig:churchgraph} \cite[Example 5]{EspiritoSantoMatthesPintoCoinductiveApproachAPAL}.
\begin{figure}[tb]
 \begin{tikzpicture}[inner sep=2mm,
         place/.style={circle,draw=blue!50,fill=blue!20,thick},
         scale=1.2]
         \node[place] at (-6,0) (root)  {};
         \node[place] at (-4,0) (lambdaf)  {$\lambda f^{0\impl0}$};
          \node[place] at (-2,0) (lambdax)  {$\lambda x^0$};
          \node[place] at (0,0) (sum)  {$+$};
          \node[place] at (-1,-1.5) (x)  {$x$};
          \node[place] at (1,-1.5) (f)  {$f@$};
          \draw [->] (root) to [out=0,in=180] (lambdaf);
          \draw [->] (lambdaf) to [out=0,in=180] (lambdax);
          \draw [->] (lambdax) to [out=0,in=180] (sum);
          \draw [->] (sum) to [out=225,in=90,looseness=0.8] (x);
          \draw [->] (sum) to [out=315,in=90,looseness=0.8] (f);
          \draw [->, blue, thick] (f) to [out=315,in=45,looseness=1.5] (sum);
         \end{tikzpicture}
\vspace{-0.4cm}
\caption{Forest representation of all Church numerals, including infinity}\label{fig:churchgraph}
\vspace{-0.4cm}
\end{figure}
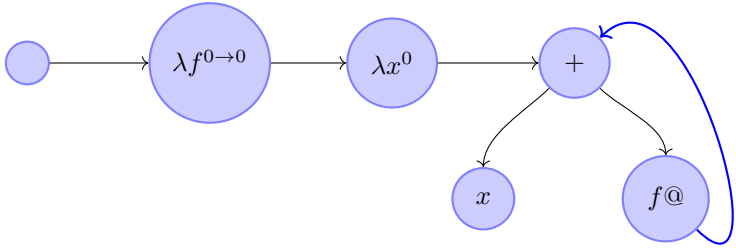
$f@$ is short for $f\tuple N$ with $N$ given by where the arrow points to.
The back link (in blue and thick) forms a cycle that does not go through a $\lambda$-abstraction.
Hence, we get a rational tree (with only a finite number of non-isomorphic subtrees).
This is a representation of all Church numerals, including infinity, and they are all the ``solutions'' (including the non-wellfounded ones) for the search for inhabitants in long normal form of the type $(0\to 0)\to 0\to 0$.
Thus, by infinite unfolding of the just binary choice, an infinite number of finite solutions and even one infinite solution are obtained.
The latter is the only infinite Church numeral, obtained by looping with $f$.
In naive proof search, this is at least a potential outcome, and one may want to analyze this phenomenon.

\section{Recalling Some Notions of the Category-theoretical Background}\label{sec:appbackground}

\fscdarxiv{\input{appendixbackgroundplaceholder}}{To complete the definitions recalled from the literature in \cref{sec:moncatsacts}, \cref{fig:moncatlaws} shows the two coherence conditions of a monoidal category while
\cref{fig:actegorylaws} shows those conditions of an actegory.

\begin{figure}[tb]
\[
  \begin{tikzcd}
    (x\otimes I)\otimes y \ar[rr, "\alpha_{x,I,y}"] \ar[dr, "\rho_x\otimes 1_y"'] & & x \otimes (I \otimes y) \ar[dl, "1_x\otimes \lambda_y"] \\
    & x \otimes y
  \end{tikzcd}
\begin{tikzcd}[column sep=-1em]
  (w\otimes x)\otimes(y\otimes z) \ar[rr, "\alpha_{w,x,y\otimes z}"] & & w\otimes (x\otimes(y\otimes z))
  \\
  ((w\otimes x)\otimes y)\otimes z \ar[u, "\alpha_{w\otimes x,y,z}"'] \ar[dr, "\alpha_{w,x,y}\otimes 1_z"'] & &w\otimes((x\otimes y)\otimes z) \ar[u, "1_w\otimes\alpha_{x,y,z}"]
  \\
    &(w\otimes(x\otimes y))\otimes z \ar[ur, "\alpha_{w,x\otimes y,z}"']
  \end{tikzcd}
\]

\caption{Triangle and pentagon law of a monoidal category}\label{fig:moncatlaws}
\end{figure}

\begin{figure}[tb]
  \[
    \begin{tikzcd}
      (v\otimes I)\odot y \ar[rr, "\actor_{v,I,y}"] \ar[dr, "\rho_v\odot 1_y"'] & & v \odot (I \odot y) \ar[dl, "1_v\odot \lambda_y"] \\
      & v \odot y
    \end{tikzcd}
    \begin{tikzcd}[column sep=-1em]
      (v\otimes w)\odot(v'\odot z) \ar[rr, "\actor_{w,v,v'\odot z}"] & & w\odot (v\odot(v'\odot z)) \\
      ((w\otimes v)\otimes v')\odot z \ar[u, "\actor_{w\otimes v,v',z}"'] \ar[dr, "\alpha_{w,v,v'}\odot 1_z"'] & &w\odot((v\otimes v')\odot z) \ar[u, "1_w\odot\actor_{v,v',z}"] \\
      &(w\otimes(v\otimes v'))\odot z \ar[ur, "\actor_{w,v\otimes v',z}"']
    \end{tikzcd}
  \]

\caption{Triangle and pentagon law of an actegory}\label{fig:actegorylaws}
\end{figure}

\cref{fig:monoid} shows the laws of a monoid in a monoidal category (\cf \coqfile{}{CategoryTheory.Monoidal.CategoriesOfMonoids}).
\begin{figure}[tb]
\[
  \begin{tikzcd}[column sep=5em]
    I\otimes v\ar[r,"\eta\otimes 1_v"] \ar[dr,"\lambda_v"']&v\otimes v\ar[d,"\mu"] & v\otimes I\ar[l,"1_v\otimes \eta"'] \ar[dl,"\rho_v"]\\
    & v
  \end{tikzcd}\hspace{1em}
  \begin{tikzcd}[column sep=1em]
    (v\otimes v)\otimes v\ar[rr,"{\alpha_{v,v,v}}"] \ar[dr,"\mu\otimes 1_v"'] && v\otimes(v\otimes v) \ar[rr,"1_v\otimes\mu"] && v\otimes v \ar[dl,"\mu"]\\
    & v\otimes v\ar[rr,"\mu"] && v
  \end{tikzcd}
\]
\caption{Laws of a $\VV$-monoid: left unit law, right unit law and associative law}\label{fig:monoid}
\end{figure}

}

\section{Proof of \cref{thm:mhssfromfinal}}\label{sec:proofmhssfromfinal}

This short appendix completes the proof of \cref{thm:mhssfromfinal} in  \cref{sec:finalcoalgtomhss}. It also discusses why the result on final coalgebras is in some sense easier than its counterpart for initial algebras.

The diagram describing a solution of $\teqm$ is given in \cref{fig:soleqmformhss}.
\begin{figure}[tb]
\[
  \begin{tikzcd}[column sep=7em, row sep=scriptsize]
    z\otimes t \ar[ddd, "h"']& z\otimes (I+Ht) \ar[l, "1_z\otimes\tout^{-1}"]& z\otimes I + z \otimes Ht \ar[l, "{[1_z\otimes\tinl,1_z\otimes\tinr]}"] \ar[d, "{\rho_z+\theta_{(z,e),t}}"]\\
    && z + H(z\otimes t) \ar[d, "{[f\cdot\tinr,\tinr\cdot\tinl]}"]\\
    &&(I+H(z\otimes t))+t\ar[d, "(I+H{-})h +1_t"]\\
    t&&(I+Ht)+t\ar[ll, "{[\tout^{-1},1_t]}"]
  \end{tikzcd}
\]
\caption{Diagram characterizing a solution $h$ for $\teqm$}\label{fig:soleqmformhss}
\end{figure}
Of course, the diagram governing $\gst{(z,e)}{f}$ can be brought into a single equation over coproducts, as seen in \cref{fig:mhssoneequation}.
\begin{figure}[tb]
\[
    \begin{tikzcd}[column sep=7em, row sep=scriptsize]
    z\otimes t \ar[dd, "h"']
    &
    z\otimes I + z \otimes Ht \ar[l, "{[1_z\otimes\eta,1_z\otimes\tau]}"]
    \ar[d, "\rho_z+\theta_{(z,e),t}"]
    \\
    & z + H(z\otimes t) \ar[d, "1_z+H h"]
    \\
    t
    &
    z+Ht \ar[l, "{[f,\tau]}"']
  \end{tikzcd}
\]
\caption{Diagram characterizing morphism $h$ that should be $\gst{(z,e)}{f}$}\label{fig:mhssoneequation}
\end{figure}
The chain of morphisms from $z\otimes I + z \otimes Ht$ to $t$ in both diagrams -- \cref{fig:soleqmformhss} and \cref{fig:mhssoneequation} -- is identical on the path to the left, as well as on the path to the right.

We remark that \cref{thm:mhssfromfinal} (and also its proof) is slicker than the case of wellfounded syntax studied in \cite{CPP22} (however, concretely for endofunctor categories) where extra requirements beyond being an initial algebra come into play so as to guarantee the applicability of a categorical Mendler-style recursion scheme.
This difference can be motivated as follows: substitution for functor $T$ is represented by a monadic multiplication operation of type $T\cdot T\to T$ (with $T\cdot T$ self-composition of $T$). In the non-wellfounded case, this has the support $T$ of the final coalgebra as target, which is suitable for using finality.
However, in the wellfounded case, the source type is $T\cdot T$ and not just $T$ that would be the basis for using initiality.

\section{Pointed Tensorial Strength for the Signature Functor}\label{sec:constrptdstrength}

\fscdarxiv{\input{appendixstrengthconstructionplaceholder}}{%

In \cref{sec:everything}, we need to construct a pointed tensorial strength $\theta$ for the signature functor $H$, associated with a given multi-sorted binding signature. We first construct ``relative lax commutators'' for the basic building blocks of $H$ in \cref{sec:appcommutator} and then use those to build  $\theta$ from pointed tensorial strength for the building blocks of $H$ that are already endofunctors on $[\ccsort,\ccsort]$.

\subsection{Relative Lax Commutators}\label{sec:appcommutator}

\begin{definition}[Relative lax commutator, \coqident{CategoryTheory.Actegories.ConstructionOfActegoryMorphisms}{relativelaxcommutator}]\label{def:commutator}
  Let $\VV$, $\WW$ be monoidal categories (the components of $\WW$ shall be indicated through primed identifiers) and $(F,\epsilon,\mu):\WW\to\VV$ be a strong monoidal functor.
  Fix a $v_0:\VV$. A \emph{lax commutator for $v_0$ relative to $F$} is given by a natural transformation $\gamma=(\gamma_w)_{w:\WW}$ with components $\gamma_w:F w\otimes v_0\to v_0\otimes F w$, subject to the unit and the tensor laws in \cref{fig:unitlawcommutator} and \cref{fig:tensorlawcommutator}, respectively.
\end{definition}

\begin{figure}[tb]
\[
  \begin{tikzcd}[column sep=3em]
    FI'\otimes v_0\ar[rr,"\gamma_{I'}"] \ar[d,"\epsilon^{-1}\otimes 1_{v_0}"']& & v_0 \otimes FI'\\
    I\otimes v_0 \ar[r,"\lambda_{v_0}"] & v_0 \ar[r,"\rho^{-1}_{v_0}"] & v_0\otimes I \ar[u,"1_{v_0}\otimes \epsilon"']
  \end{tikzcd}
\]
\caption{Unit law of a lax commutator $\gamma$ for $v_0$ relative to $F$}\label{fig:unitlawcommutator}
\end{figure}

\begin{figure}[tb]
\[
  \begin{tikzcd}[column sep=3em]
    F(w\otimes'w')\otimes v_0 \ar[rrr,"\gamma_{w\otimes w'}"] \ar[d,"\mu^{-1}_{w,w'}\otimes 1_{v_0}"']&&&v_0\otimes F(w\otimes' w')\\
    (Fw\otimes Fw')\otimes v_0\ar[d,"\alpha"']&&&v_0\otimes (Fw\otimes Fw')\ar[u,"1_{v_0}\otimes\mu_{w,w'}"']\\
    Fw \otimes (Fw'\otimes v_0)\ar[r,"1_{Fw}\otimes \gamma_{w'}"]&
    Fw\otimes (v_0\otimes Fw')\ar[r,"\alpha^{-1}"]&
    (Fw \otimes v_0)\otimes Fw'\ar[r,"\gamma_w\otimes1_{Fw'}"]&
    (v_0\otimes Fw)\otimes Fw'\ar[u,"\alpha"']
  \end{tikzcd}
\]
\caption{Tensor law of a lax commutator $\gamma$ for $v_0$ relative to $F$. (The indices of the associator of $\VV$ are omitted.)}\label{fig:tensorlawcommutator}
\end{figure}
The unit law uniquely determines $\gamma_{I'}$, and the tensor law uniquely determines $\gamma_{w\otimes w'}$ when $\gamma_w$ and $\gamma_{w'}$ are already given.

In the case that $(F,\epsilon,\mu)$ is the identity on $\VV$, we can omit the vertical arrows in the unit law and the uppermost vertical arrows in the tensor law.
The tensor law then becomes one of the two hexagonal laws of a braided monoidal category. But beware that we fix the object $v_0$.
In a \emph{symmetric} monoidal category, the shrunk down unit law (for $F$ the identity) is even derivable.\footnote{Remark 2.4 in nLab entry \url{https://ncatlab.org/nlab/revision/braided+monoidal+category/51}.}
So, in essence, this is a braiding, but it is relative to that strong monoidal functor $F$, there is no isomorphism condition, and since $v_0$ is fixed, only half of the laws make sense.
(And the unit law has to be stated explicitly to compensate for that.)\footnote{We also considered the name ``relative semi-braiding'' instead of ``relative lax commutator''.}

A less abstract version of the notion of relative lax commutator was introduced under the name ``pointed distributive law'' in \cite[Definition 10]{signatures_to_monads}.
It corresponds to the instance where $\VV$ is $[\CC,\CC]$, $\WW$ the corresponding pointed endofunctors and $F$ the forgetful functor that forgets the points.
In that instance, all the morphisms in the unit law diagram are (pointwise) the identity on $v_0$, as well as the preservation of the tensor ($\mu$ and $\mu^{-1}$) in the tensor law diagram.
In hindsight, ``pointed distributive law'' rather appears as a misnomer.

However, the results in \cite{signatures_to_monads} smoothly generalize to our relative lax commutator.
To begin with, the abstract counterpart to \cite[Lemma 11]{signatures_to_monads} is as follows: 
\begin{construction}[Lineator for reindexed actegory from relative lax commutator]\label{constr:lineatorfromcommutator}
  Let $\VV$, $\WW$ and $F$ be as in \cref{def:commutator}.
  Assume a $\VV$-actegory with action $\odot$ on a category $\CC$.
  Fix some $v_0:\VV$ and consider the endofunctor $G$ on $\CC$ given by $G\eqdef 1_{v_0}\odot{-}$ (the left whiskering functor for $v_0$).
  Let $\odot'$ be the action of the $\WW$-actegory obtained by reindexing $\odot$ along $F$.
  Let $\gamma$ be a lax commutator for $v_0$ relative to $F$.
  Then we construct a lax lineator $\ell$ for $G$ with $\odot'$ as source and target, by the (pointwise) definition of $\ell_{w,x}$ given by the following diagram:%
  \[
    \begin{tikzcd}[column sep=5em]
      Fw \odot(v_0\odot x) \ar[r,"\ell_{w,x}"] \ar[d,"\actor^{-1}_{Fw,v_0,x}"']& v_0\odot(Fw\odot x)\\
      (Fw \otimes v_0) \odot x\ar[r,"\gamma_w\odot 1_x"] & (v_0\otimes Fw)\odot x \ar[u,"\actor_{v_0,Fw,x}"']
    \end{tikzcd}
  \]
 \end{construction}
\begin{lemma}
  The thus constructed $\ell$ is indeed a lax lineator for $G$ with $\odot'$ as source and target.
\end{lemma}
\begin{proof}
  The proof of the laws has been formalized (\coqident[reindexedlineator_from_commutator]{CategoryTheory.Actegories.ConstructionOfActegoryMorphisms}{reindexedstrength_from_commutator}).
\end{proof}

The abstract counterpart to \cite[Lemma 12]{signatures_to_monads} is as follows:
\begin{construction}[Composition of relative lax commutators]\label{constr:compcommutator}
  Let $\VV$, $\WW$ and $F$ be as in \cref{def:commutator}.
  Let furthermore $\gamma^i$ be a lax commutator for $v_i:\VV$ relative to $F$, for $i\in\{1,2\}$.
  Then we construct a lax commutator $\gamma$ for $v_1\otimes v_2$ relative to $F$, by the (pointwise) definition of $\gamma_w$ given by the diagram in \cref{fig:compositecommutator} (omitting the indices of the associator).
\end{construction}
\begin{figure}[tb]
  \[
    \begin{tikzcd}[column sep=3.5em]
      Fw\otimes(v_1\otimes v_2)\ar[rrr,"\gamma_w"] \ar[d,"\alpha^{-1}"'] &&& (v_1\otimes v_2)\otimes Fw \\
      (Fw \otimes v_1)\otimes v_2\ar[r,"\gamma^1_w\otimes 1_{v_2}"] & (v_1\otimes Fw)\otimes v_2 \ar[r,"\alpha"] & v_1\otimes (Fw \otimes v_2) \ar[r,"1_{v_1}\otimes\gamma^2_w"] & v_1\otimes(v_2\otimes Fw) \ar[u,"\alpha^{-1}"']
    \end{tikzcd}
  \]
  \caption{Lax commutator $\gamma$ for $v_1\otimes v_2$ relative to $F$.}\label{fig:compositecommutator}
\end{figure}

\begin{lemma}
  The thus constructed $\gamma$ is indeed a lax commutator $\gamma$ for $v_1\otimes v_2$ relative to $F$.
\end{lemma}
\begin{proof}
  The proof of the laws has been formalized (\coqident{CategoryTheory.Actegories.ConstructionOfActegoryMorphisms}{composedrelativelaxcommutator}).
\end{proof}

We also construct a lax commutator for the unit relative to any $F$.
For this (\cf \coqident{CategoryTheory.Actegories.ConstructionOfActegoryMorphisms}{unit_relativelaxcommutator}), we use that, in a monoidal category, the left and right unitor agree at the unit.

We now obtain an essential building block for the construction of a pointed tensorial strength for the signature functor $H$ defined in \cref{sec:signaturefunctor} (given a multi-sorted binding signature).
\begin{construction}\label{constr:strengthprecompwithoptionstar}
  Given a set $\sort$ of sorts and a list $\ell$ of sorts, we construct a pointed tensorial strength $\theta_0^\ell$ for the functor given by precomposition with $\option^*~\ell$.

  As a first step, for any sort $s$, we obtain a lax commutator $\gamma^s$ for $\option_s$ relative to the forgetful functor $U$ used for the definition of pointed tensorial strength, \cf \coqident{SubstitutionSystems.BindingSigToMonad_actegorical}{ptdlaxcommutator_genopt}.
  To recall the task: given a pointed endofunctor $(Y,\eta)$, \ie, $Y:\ccsort\to\ccsort$ and $\eta:\Id{}\to Y$, one has to define $\gamma^s_{(Y,\eta)}:Y\cdot\option_s\to \option_s\cdot Y$, which is done pointwise: $\gamma^s_{(Y,\eta),\xi}:\option_s(Y\xi)\to Y(\option_s\,\xi)$.
  This is again defined pointwise: $\gamma^s_{(Y,\eta),\xi,t}:\option_s(Y\xi)t\to Y(\option_s\,\xi)t$.
  One can now do case analysis on the source coproduct.
  In the left branch, one uses $\eta$\footnote{This is the well-known reason why one has to work with \emph{pointed} tensorial strength when dealing with variable binding.}, in the right branch, one applies functoriality of $Y$ -- this all analogously to the well-known un(i)sorted case. The proof that this indeed satisfied the laws of \cref{def:commutator} is routine.

  As a second step, using $\gamma^s$ for all the sorts $s$ in our list $\ell$, we get a lax commutator $\gamma^\ell$ for $\option^*_\ell$ relative to $U$ (\cf\coqident{SubstitutionSystems.MultiSorted_actegorical}{ptdlaxcommutatorCAT_option_list}), by induction on the length of $\ell$:
  If $\ell$ is empty, we use the remark preceding this lemma about the construction of a lax commutator for the unit $\Id{}$.
  The composite case is a direct use of \cref{constr:compcommutator} and the inductive hypothesis.

  In the final step, we apply \cref{constr:lineatorfromcommutator} with the canonical self action as action $\odot$, etc. -- the precise situation of the definition of pointed tensorial strength through a reindexed actegory.
\end{construction}

\subsection{Using Actegory Theory}\label{sec:usingactegorytheory}

Let $\sort$ be a set of sorts and $(I,\arity)$ a multi-sorted binding signature over $\sort$.
We now construct a pointed tensorial strength $\theta$ for $H$, as defined in \cref{sec:signaturefunctor}.
We construct the strength from strengths\footnote{We will from now on more briefly speak of ``strength'' instead of ``pointed tensorial strength''.} for ``building blocks'', that is, for functors $F: [\ccsort,\ccsort] \to [\ccsort,\ccsort]$ from which the signature functor $H$ is obtained.
For the more basic building blocks that are just endofunctors on $\ccsort$, such as $\option^*_\ell$, we already have \cref{constr:strengthprecompwithoptionstar}.

We follow the top-down approach of \cref{sec:signaturefunctor} and always explain the general results on actegories that we instantiate to endofunctor categories (which makes the method different from that in \cite{CPP22}).

We start with the final step of the construction (\coqident{SubstitutionSystems.MultiSorted_actegorical}{MultiSortedSigToStrength'}).
Assume that we already have a strength $\theta^{(\vec a,t)}$ for all $G^{(\vec a,t)}$, with $(\vec a,t):(\sort^* \times \sort)^* \times \sort$.
In general, given actegories on $\CC$ and $\DD$ and an $I$-indexed family of functors from $\CC$ to $\DD$ with lineators, there is a construction (\coqident{CategoryTheory.Actegories.ConstructionOfActegoryMorphisms}{lax_lineator_coprod}) of a lineator for the coproduct of that family of functors under the proviso of having distribution of $I$-indexed coproducts over the actegory on $\DD$.
Analogously to what we needed at the beginning of \cref{sec:finalcoalgtomhss}, this means that the morphism from $\coprod_{i:I}(v\odot' x_i)$ to $v\odot'\coprod_{i:I}x_i$, defined by case analysis and injections, has an inverse.
Such distribution is inherited through reindexing of actions (\coqident{CategoryTheory.Actegories.CoproductsInActegories}{reindexed_coprod_distributor}), and it is available in the canonical self-action of an endofunctor category (given $I$-indexed coproducts in the base category).
Hence, the lineators/strengths $\theta^{\arity(i)}$ for $i:I$ can be put together to get a lineator/strength for $H$.

The penultimate step (\coqident{SubstitutionSystems.MultiSorted_actegorical}{StrengthCAT_hat_exp_functor_list'_optimized}) constructs strength $\theta^{(\vec a,t)}$ for $G^{(\vec a,t)}$.
It assumes that we already have a strength $\theta^{(a,t)}$ for any $G^{(a,t)}$, with $(a,t):(\sort^* \times \sort)\times \sort$.

In general, there is a construction that reindexes lineators: in the situation of the definition of a reindexed actegory (\cref{sec:ptdstrength}), we assume a further $\VV$-actegory on category $\DD$, a functor $G:\CC\to\DD$ and a lineator $\ell$ for $G$ between those actegories.
Given $\ell_{v,x}:v\odot' Gx\to G(v\odot x)$, we simply set $\ell'_{w,x}\eqdef \ell_{Fw,x}$ and obtain a lineator $\ell'$ for the same functor $G$ but between the $\WW$-actegories obtained from reindexing the given actegories on $\CC$ and $\DD$ along $F$, respectively (\coqident{CategoryTheory.Actegories.ConstructionOfActegoryMorphisms}{reindexed_lax_lineator}).

We now detail the construction of $\theta^{([],t)}$. Thanks to reindexing, we only need to construct a lineator with respect to the canonical self-actions.
Since $G^{([],t)}$ has been presented as a composition of two functors, we only need to compose the lineators for both (thanks to the generic construction of the lineator of a composition from lineators of the constituents which is a building block of the well-known bicategory of actegories and lax functors).
There is in general a lineator for constantly constant functors $F:\DD\to[\CC,\EE]$ with an arbitrary $[\CC,\CC]$-actegory on $\DD$ as source, and with the $[\CC,\CC]$-actegory on $[\CC,\EE]$ in the target given by precomposition with endofunctors on $\CC$ -- the lineator is pointwise the identity. (Notice that, in general, constant functors do not have lineators.)
We thus have a lineator for the first functor in the presentation of $G^{([],t)}$.

\begin{construction}[Lineator for post-composition with fixed functor]\label{constr:lineatorpostcomp}
  We construct a lineator $\ell$ for post-composition with a fixed functor $G:\DD\to\CC$, with the source $[\CC,\CC]$-actegory on $[\CC,\DD]$ given by precomposition, and the target $[\CC,\CC]$-actegory the one with the canonical self-action.
For $F:\CC\to\CC$ and $K:\CC\to\DD$, $\ell_{F,K}$ is just the associator $\alpha:F\cdot(K\cdot G)\to (F\cdot K)\cdot G$.
\end{construction}
\begin{lemma}
  The thus defined family $(\ell_{F,K})_{F,K}$ is indeed a lineator for the actegories announced in the construction.
\end{lemma}
\begin{proof}
  The proof has been formalized (\coqident{CategoryTheory.Actegories.Examples.SelfActionInCATElementary}{lax_lineator_postcomp_SelfActCAT_alt}).
\end{proof}
We thus also have a lineator for the second functor in the presentation of $G^{([],t)}$.

We come to the case of nonempty $\vec a$ in the construction of $\theta^{(\vec a,t)}$. Concerning the pointwise binary product of functors, their lineators combine without any extra assumption (other than the existence of the binary products in the target category needed to formulate the pointwise binary product, \cf\coqident{CategoryTheory.Actegories.ConstructionOfActegoryMorphisms}{lax_lineator_coprod}). Therefore, a simple induction on the length of $\vec a$ concludes this construction.

It remains to construct strength $\theta^{(a,t)}$.
In view of the presentation of $G^{(a,t)}$ as composition of two functors and the above-mentioned fact that lineators compose, we only have to construct a strength for each component. The first one is just $\theta_0^\ell$ (\cref{constr:strengthprecompwithoptionstar}) for the functor given by precomposition with $\option^*~\ell$. Therefore it remains to construct a strength for postcomposition with $\proj{s} \hcomp \hat t$.
We do this analogously to the second component of $G^{([],t)}$: using reindexing, it suffices to construct a lineator for the involved $[\CC,\CC]$-actegories, but this is just another instance of \cref{constr:lineatorpostcomp}.

}

\section{On the Formalization}\label{sec:formal}

Most of the definitions and results presented in this paper are formalized and computer-checked in \UniMath~\cite{UniMath}, a library of univalent mathematics based on the computer proof assistant \Coq~\cite{coq}.
An exception is the application scenario in \cref{sec:intro};  its formalization is ongoing work. 
For this application, we can only offer the instantiation of the general constructions of this paper but not yet the inhabitation analysis alluded to in \cref{fig:threegraph}.
Our HTML documentation is derived from commit \href{https://github.com/UniMath/UniMath/tree/\longhash}{\shorthash} of the \UniMath library.
Proof-checking and creation of the HTML documentation can easily be reproduced at home by following the \UniMath compilation instructions -- do try this at home!

Concerning coinductive definitions, \Coq features a built-in mechanism for specifying coinductive types (via the keyword \verb|CoInductive|) and for defining functions by corecursion.
However, the \UniMath library departs from standard use of \Coq in that such declarations of coinductive datatypes  are not part of the language used in \UniMath.
Furthermore, definitions by corecursion in \Coq face numerous issues with guardedness, in particular with so-called ``mixed inductive-coinductive'' declarations \cite{BasoldPhD} -- declarations where the coinductive type makes use of a parameterized inductive type whose parameter is built with the coinductive type.
The coinductive calculus of our application scenario (see \cref{sec:intro}), with its lists of alternatives and arguments, falls into that class.
In the formalization of the contents of this paper, we therefore construct coinductive datatypes from other type constructors, rather than postulating (a class of) coinductive datatypes using meta-theoretic devices.
Our approach is thus comparable to the one employed for working with coinductive datatypes in the Isabelle system \cite{DBLP:journals/pacmpl/BlanchetteGPT19} and of the construction of indexed M-types in univalent foundations~\cite{DBLP:conf/tlca/AhrensCS15}; in all of these cases, a major goal is to keep the ``trusted code base'' small.

We now discuss some design choices we made in the formalization.
When formalizing mathematics in a formal system, some design choices need to be made that are not of mathematical significance: different choices lead to (trivially) equivalent mathematical concepts.
Nevertheless, making the right choices can be crucial for the maintainability and usability of the formal library.
An example of such a choice is the following.
In \cref{sec:moncatsacts}, we said that a monoidal category is given by a six-tuple,  with the tensor component a bifunctor.
However, for the sake of our formalization, we have chosen a different but equivalent format to present the tensor operation that we are calling ``whiskered''.
Here, the object mapping of $\otimes$ is replaced by its curried version, and the morphism mapping is replaced by two families of endofunctors on $\CC$ that represent the morphism mappings with one of the arguments fixed to the identity morphism -- thus the ``whiskerings'' of that bifunctor.
The whiskered definition avoids functors on cartesian products of categories.
Such functors do not behave well in practice:
the inference of the implicit object arguments -- which are pairs of objects -- to the functorial action on morphisms often fails, and thus these arguments need to be given explicitly.
This would make the formalization cumbersome -- which is why we adopted the whiskered format for our work and hence do not suffer from those problems.
A third alternative to the traditional definition and the whiskered definition would be a currying of the tensor, to be a functor into a functor category.
However, this definition would not provide a clean separation between data and properties -- another prerequisite for a library that scales well, in our experience.

\fscdarxiv{\bibliography{literatureFSCD24}}{}

\begin{thebibliography}{10}

\bibitem{DBLP:conf/tlca/AhrensCS15}
Benedikt Ahrens, Paolo Capriotti, and R{\'{e}}gis Spadotti.
\newblock Non-wellfounded trees in homotopy type theory.
\newblock In Thorsten Altenkirch, editor, {\em 13th International Conference on
  Typed Lambda Calculi and Applications, {TLCA} 2015, July 1-3, 2015, Warsaw,
  Poland}, volume~38 of {\em LIPIcs}, pages 17--30. Schloss Dagstuhl -
  Leibniz-Zentrum f{\"{u}}r Informatik, 2015.
\newblock \href {https://doi.org/10.4230/LIPIcs.TLCA.2015.17}
  {\path{doi:10.4230/LIPIcs.TLCA.2015.17}}.

\bibitem{DBLP:conf/rta/AhrensHLM19}
Benedikt Ahrens, Andr{\'{e}} Hirschowitz, Ambroise Lafont, and Marco Maggesi.
\newblock Modular specification of monads through higher-order presentations.
\newblock In Herman Geuvers, editor, {\em 4th International Conference on
  Formal Structures for Computation and Deduction, {FSCD} 2019, June 24-30,
  2019, Dortmund, Germany}, volume 131 of {\em LIPIcs}, pages 6:1--6:19.
  Schloss Dagstuhl - Leibniz-Zentrum f{\"{u}}r Informatik, 2019.
\newblock \href {https://doi.org/10.4230/LIPIcs.FSCD.2019.6}
  {\path{doi:10.4230/LIPIcs.FSCD.2019.6}}.

\bibitem{DBLP:journals/pacmpl/AhrensHLM20}
Benedikt Ahrens, Andr{\'{e}} Hirschowitz, Ambroise Lafont, and Marco Maggesi.
\newblock Reduction monads and their signatures.
\newblock {\em Proc. {ACM} Program. Lang.}, 4({POPL}):31:1--31:29, 2020.
\newblock \href {https://doi.org/10.1145/3371099} {\path{doi:10.1145/3371099}}.

\bibitem{signatures_to_monads}
Benedikt Ahrens, Ralph Matthes, and Anders M{\"{o}}rtberg.
\newblock From signatures to monads in {UniMath}.
\newblock {\em J. Autom. Reason.}, 63(2):285--318, 2019.
\newblock \href {https://doi.org/10.1007/s10817-018-9474-4}
  {\path{doi:10.1007/s10817-018-9474-4}}.

\bibitem{CPP22}
Benedikt Ahrens, Ralph Matthes, and Anders M{\"{o}}rtberg.
\newblock Implementing a category-theoretic framework for typed abstract
  syntax.
\newblock In Andrei Popescu and Steve Zdancewic, editors, {\em {CPP} '22: 11th
  {ACM} {SIGPLAN} International Conference on Certified Programs and Proofs,
  Philadelphia, PA, USA, January 17 - 18, 2022}, pages 307--323. {ACM}, 2022.
\newblock \href {https://doi.org/10.1145/3497775.3503678}
  {\path{doi:10.1145/3497775.3503678}}.

\bibitem{DBLP:journals/jfp/AllaisACMM21}
Guillaume Allais, Robert Atkey, James Chapman, Conor McBride, and James
  McKinna.
\newblock A type- and scope-safe universe of syntaxes with binding: their
  semantics and proofs.
\newblock {\em J. Funct. Program.}, 31:e22, 2021.
\newblock \href {https://doi.org/10.1017/S0956796820000076}
  {\path{doi:10.1017/S0956796820000076}}.

\bibitem{DBLP:conf/csl/AltenkirchR99}
Thorsten Altenkirch and Bernhard Reus.
\newblock Monadic presentations of lambda terms using generalized inductive
  types.
\newblock In J{\"{o}}rg Flum and Mario Rodr{\'{\i}}guez{-}Artalejo, editors,
  {\em Computer Science Logic, 13th International Workshop, {CSL} '99, 8th
  Annual Conference of the EACSL, Madrid, Spain, September 20-25, 1999,
  Proceedings}, volume 1683 of {\em Lecture Notes in Computer Science}, pages
  453--468. Springer, 1999.
\newblock \href {https://doi.org/10.1007/3-540-48168-0\_32}
  {\path{doi:10.1007/3-540-48168-0\_32}}.

\bibitem{lambdacalculuswithtypes}
Hendrik~Pieter Barendregt, Wil Dekkers, and Richard Statman.
\newblock {\em Lambda Calculus with Types}.
\newblock Perspectives in logic. Cambridge University Press, 2013.
\newblock URL:
  \url{http://www.cambridge.org/de/academic/subjects/mathematics/logic-categories-and-sets/lambda-calculus-types}.

\bibitem{BasoldPhD}
Henning Basold.
\newblock {\em Mixed Inductive-Coinductive Reasoning---Types, Programs and
  Logic}.
\newblock PhD thesis, Radboud University, Nijmegen, The Netherlands, 2018.
\newblock URL: \url{https://repository.ubn.ru.nl/handle/2066/190323}.

\bibitem{birdmeertens}
Richard Bird and Lambert Meertens.
\newblock {Nested Datatypes}.
\newblock In Johan Jeuring, editor, {\em Mathematics of {P}rogram
  {C}onstruction, {MPC}'98, Proceedings}, volume 1422 of {\em Lecture Notes in
  Computer Science}, pages 52--67. Springer, 1998.

\bibitem{DBLP:journals/jfp/BirdP99}
Richard~S. Bird and Ross Paterson.
\newblock {De Bruijn Notation as a Nested Datatype}.
\newblock {\em J. Funct. Program.}, 9(1):77--91, 1999.
\newblock URL:
  \url{http://journals.cambridge.org/action/displayAbstract?aid=44239}.

\bibitem{DBLP:journals/pacmpl/BlanchetteGPT19}
Jasmin~Christian Blanchette, Lorenzo Gheri, Andrei Popescu, and Dmitriy
  Traytel.
\newblock Bindings as bounded natural functors.
\newblock {\em Proc. {ACM} Program. Lang.}, 3({POPL}):22:1--22:34, 2019.
\newblock \href {https://doi.org/10.1145/3290335} {\path{doi:10.1145/3290335}}.

\bibitem{DBLP:conf/lics/BorthelleHL20}
Peio Borthelle, Tom Hirschowitz, and Ambroise Lafont.
\newblock A cellular {Howe} theorem.
\newblock In Holger Hermanns, Lijun Zhang, Naoki Kobayashi, and Dale Miller,
  editors, {\em {LICS} '20: 35th Annual {ACM/IEEE} Symposium on Logic in
  Computer Science, Saarbr{\"{u}}cken, Germany, July 8-11, 2020}, pages
  273--286. {ACM}, 2020.
\newblock \href {https://doi.org/10.1145/3373718.3394738}
  {\path{doi:10.1145/3373718.3394738}}.

\bibitem{Actegories}
Matteo Capucci and Bruno Gavranović.
\newblock Actegories for the working amthematician, 2022.
\newblock \href {https://doi.org/10.48550/ARXIV.2203.16351}
  {\path{doi:10.48550/ARXIV.2203.16351}}.

\bibitem{DBLP:journals/pacmpl/FioreS22}
Marcelo Fiore and Dmitrij Szamozvancev.
\newblock Formal metatheory of second-order abstract syntax.
\newblock {\em Proc. {ACM} Program. Lang.}, 6({POPL}):1--29, 2022.
\newblock \href {https://doi.org/10.1145/3498715} {\path{doi:10.1145/3498715}}.

\bibitem{DBLP:conf/lics/Fiore08}
Marcelo~P. Fiore.
\newblock Second-order and dependently-sorted abstract syntax.
\newblock In {\em Proceedings of the Twenty-Third Annual {IEEE} Symposium on
  Logic in Computer Science, {LICS} 2008, 24-27 June 2008, Pittsburgh, PA,
  {USA}}, pages 57--68. {IEEE} Computer Society, 2008.
\newblock \href {https://doi.org/10.1109/LICS.2008.38}
  {\path{doi:10.1109/LICS.2008.38}}.

\bibitem{DBLP:conf/lics/FiorePT99}
Marcelo~P. Fiore, Gordon~D. Plotkin, and Daniele Turi.
\newblock Abstract syntax and variable binding.
\newblock In {\em 14th Annual {IEEE} Symposium on Logic in Computer Science,
  Trento, Italy, July 2-5, 1999}, pages 193--202. {IEEE} Computer Society,
  1999.
\newblock \href {https://doi.org/10.1109/LICS.1999.782615}
  {\path{doi:10.1109/LICS.1999.782615}}.

\bibitem{DBLP:conf/fossacs/HirschowitzHLM22}
Andr{\'{e}} Hirschowitz, Tom Hirschowitz, Ambroise Lafont, and Marco Maggesi.
\newblock Variable binding and substitution for (nameless) dummies.
\newblock In Patricia Bouyer and Lutz Schr{\"{o}}der, editors, {\em Foundations
  of Software Science and Computation Structures - 25th International
  Conference, {FOSSACS} 2022, Held as Part of the European Joint Conferences on
  Theory and Practice of Software, {ETAPS} 2022, Munich, Germany, April 2-7,
  2022, Proceedings}, volume 13242 of {\em Lecture Notes in Computer Science},
  pages 389--408. Springer, 2022.
\newblock \href {https://doi.org/10.1007/978-3-030-99253-8\_20}
  {\path{doi:10.1007/978-3-030-99253-8\_20}}.

\bibitem{DBLP:phd/ethos/Hur10}
Chung{-}Kil Hur.
\newblock {\em Categorical equational systems : algebraic models and equational
  reasoning}.
\newblock PhD thesis, University of Cambridge, {UK}, 2010.
\newblock URL: \url{http://ethos.bl.uk/OrderDetails.do?uin=uk.bl.ethos.608664}.

\bibitem{DBLP:journals/corr/KurzPSV13}
Alexander Kurz, Daniela Petrisan, Paula Severi, and Fer{-}Jan de~Vries.
\newblock Nominal coalgebraic data types with applications to lambda calculus.
\newblock {\em Log. Methods Comput. Sci.}, 9(4), 2013.
\newblock \href {https://doi.org/10.2168/LMCS-9(4:20)2013}
  {\path{doi:10.2168/LMCS-9(4:20)2013}}.

\bibitem{lamiaux2024introduction}
Thomas Lamiaux and Benedikt Ahrens.
\newblock An introduction to different approaches to initial semantics, 2024.
\newblock \href {http://arxiv.org/abs/2401.09366} {\path{arXiv:2401.09366}}.

\bibitem{MatthesUustaluTCS}
Ralph Matthes and Tarmo Uustalu.
\newblock Substitution in non-wellfounded syntax with variable binding.
\newblock {\em Theoretical Computer Science}, 327(1-2):155--174, 2004.
\newblock \href {https://doi.org/10.1016/j.tcs.2004.07.025}
  {\path{doi:10.1016/j.tcs.2004.07.025}}.

\bibitem{FSCD24MatthesWullaertAhrens}
Ralph Matthes, Kobe Wullaert, and Benedikt Ahrens.
\newblock Substitution for non-wellfounded syntax with binders through monoidal
  categories.
\newblock In Jakob Rehof, editor, {\em 9th International Conference on Formal
  Structures for Computation and Deduction ({FSCD} 2024), July 10--13, 2024,
  Tallinn, Estonia}, volume 299 of {\em LIPIcs}, pages 22:1--22:22. Schloss
  Dagstuhl - Leibniz-Zentrum f{\"{u}}r Informatik, 2024.
\newblock To appear.
\newblock \href {https://doi.org/10.4230/LIPIcs.FSCD.2024.22}
  {\path{doi:10.4230/LIPIcs.FSCD.2024.22}}.

\bibitem{DBLP:conf/lics/Mellies17}
Paul{-}Andr{\'{e}} Melli{\`{e}}s.
\newblock Higher-order parity automata.
\newblock In {\em 32nd Annual {ACM/IEEE} Symposium on Logic in Computer
  Science, {LICS} 2017, Reykjavik, Iceland, June 20-23, 2017}, pages 1--12.
  {IEEE} Computer Society, 2017.
\newblock \href {https://doi.org/10.1109/LICS.2017.8005077}
  {\path{doi:10.1109/LICS.2017.8005077}}.

\bibitem{DBLP:journals/iandc/Milius05}
Stefan Milius.
\newblock Completely iterative algebras and completely iterative monads.
\newblock {\em Inf. Comput.}, 196(1):1--41, 2005.
\newblock \href {https://doi.org/10.1016/j.ic.2004.05.003}
  {\path{doi:10.1016/j.ic.2004.05.003}}.

\bibitem{DBLP:journals/pacmpl/Popescu24}
Andrei Popescu.
\newblock Nominal recursors as epi-recursors.
\newblock {\em Proc. {ACM} Program. Lang.}, 8({POPL}):425--456, 2024.
\newblock \href {https://doi.org/10.1145/3632857} {\path{doi:10.1145/3632857}}.

\bibitem{EspiritoSantoMatthesPintoCoinductiveApproachAPAL}
Jos{\'{e}}~Esp{\'{\i}}rito Santo, Ralph Matthes, and Lu{\'{\i}}s Pinto.
\newblock A coinductive approach to proof search through typed lambda-calculi.
\newblock {\em Ann. Pure Appl. Log.}, 172(10):103026, 2021.
\newblock \href {https://doi.org/10.1016/j.apal.2021.103026}
  {\path{doi:10.1016/j.apal.2021.103026}}.

\bibitem{coq}
The Coq~Development Team.
\newblock The {Coq} proof assistant, version 8.17, 2023.
\newblock URL: \url{https://zenodo.org/record/8161141}.

\bibitem{UniMath}
Vladimir Voevodsky, Benedikt Ahrens, Daniel Grayson, et~al.
\newblock {UniMath --- a computer-checked library of univalent mathematics}.
\newblock {Available at \url{http://unimath.github.io/UniMath/} }, 2021.

\bibitem{Zsido-phd}
Julianna Zsid{\'o}.
\newblock {\em {Typed Abstract Syntax}}.
\newblock PhD thesis, University of Nice Sophia Antipolis, 2010.
\newblock \url{http://tel.archives-ouvertes.fr/tel-00535944/}.

\end{thebibliography}
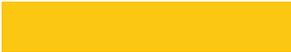
\end{document}